\documentclass[journal]{IEEEtran}
\usepackage{blindtext, graphicx}
\usepackage{amsmath,amsfonts,amssymb,bbm}
\usepackage{algorithm,algorithmic}
\usepackage{epstopdf}
\usepackage{epsfig}
\graphicspath{{images/}}
\usepackage{mathtools,dsfont}
\usepackage{psfrag,bbm,graphicx}
\usepackage{algorithm,algorithmic}
\usepackage{comment,balance}
\usepackage{epstopdf}
\usepackage{epsfig}
\usepackage{multicol}
\usepackage{multirow, cite}
\usepackage{amsfonts}
\usepackage{color}
\usepackage{lipsum}

\def \OO {\mathrm{O}}
\def \oo {\mathrm{o}}
\def\EE{{\mathbb{E}}}
\def\PP{{\mathbb{P}}}

\ifCLASSINFOpdf

\else

\fi

\newtheorem{theorem}{Theorem}

\newtheorem{lemma}{Lemma}
\newtheorem{proposition}{Proposition}

\newtheorem{corollary}[theorem]{Corollary}

\newtheorem{remark}{Remark}
\newenvironment{proof}{{\bf Proof:}}{\hfill\rule{2mm}{2mm}}
\newcommand{\remove}[1]{}
\begin{document}
	%
	\title{Resource Pooling in Large-Scale \\Content Delivery Systems}

	\author{Kota Srinivas Reddy, Sharayu Moharir and Nikhil Karamchandani \\
		Department of Electrical Engineering \\
		Indian Institute of Technology, Bombay \\
		Email: ksreddy@ee.iitb.ac.in, sharayum@ee.iitb.ac.in, nikhilk@ee.iitb.ac.in
	}
	\maketitle
	\begin{abstract}
		Content delivery networks are a key infrastructure component used by Video on Demand (VoD) services to deliver content over the Internet. We study a content delivery system consisting of a central server and multiple co-located caches, each with limited storage and service capabilities. This work evaluates the performance of such a system as a function of the storage capacity of the caches, the content replication strategy, and the service policy. This analysis can be used for a system-level optimization of these design choices. 
		
		The focus of this work is on understanding the benefits of allowing caches to pool their resources to serve user requests. We show that the benefits of resource pooling depend on the popularity profile of the contents offered by the VoD service. 
		More specifically, if the popularity does not vary drastically
		across contents, then resource pooling leads to an order wise reduction in central server transmission rate as the system size grows. On the other hand, if the content popularity is skewed, the central server transmission rate is
		of the same order with and without resource pooling.
		%
	\end{abstract}
	%
	\IEEEpeerreviewmaketitle
	\section{Introduction}\label{sec:introuction}
{\let\thefootnote\relax\footnote{Preliminary versions of this work appeared in  \cite{moharir2017content} and  \cite{reddy2017resource}.  This work was supported in part by  a SERB grant on ``Content Caching and Delivery over
		Wireless Networks" and seed grants from IIT Bombay.}}
{
The popularity of Video on Demand (VoD) services like YouTube \cite{Youtube} is ever increasing. It is predicted that VoD services will account for over 81\% of all the Internet traffic by 2022 \cite{Cisco1}. Most popular VoD services use distributed content delivery networks to serve their customers. In this work, we study a distributed content delivery network with multiple caches deployed in a geographical area (see Figure \ref{fig:cache_cluster}). Content is delivered to the users either by these caches or by a common root node, which is connected to the central server that stores the entire content catalog offered by the VoD service. As discussed in \cite{borst2010distributed, moharir2017content}, this model captures the setting where the ISP, represented by the root node, uses the distributed local caches to serve user requests and thus help reduce communication with the core network represented by the central server. This cache cluster can also be a part of a larger tree network \cite{borst2010distributed}. 

Most popular VoD services have massive content catalogs and serve a large number of users. Motivated by this, we study a time-slotted system where a batch of requests arrives in each time-slot. Each request is for a content from the  catalog offered by the VoD service. The system uses the caches to serve as many of these requests as possible, and the remaining requests are directed to the central server. The goal is to design a placement and service policy to minimize the number of contents which need to be fetched from the central server. 

The design choices in such systems include dimensioning the cache storage resources, optimizing content replication on the caches, designing policies for routing and serving user requests. This work develops a model to enable a system-level optimization of these design choices. 

\begin{figure}[t]
	\begin{center}
		\includegraphics[scale=0.33]{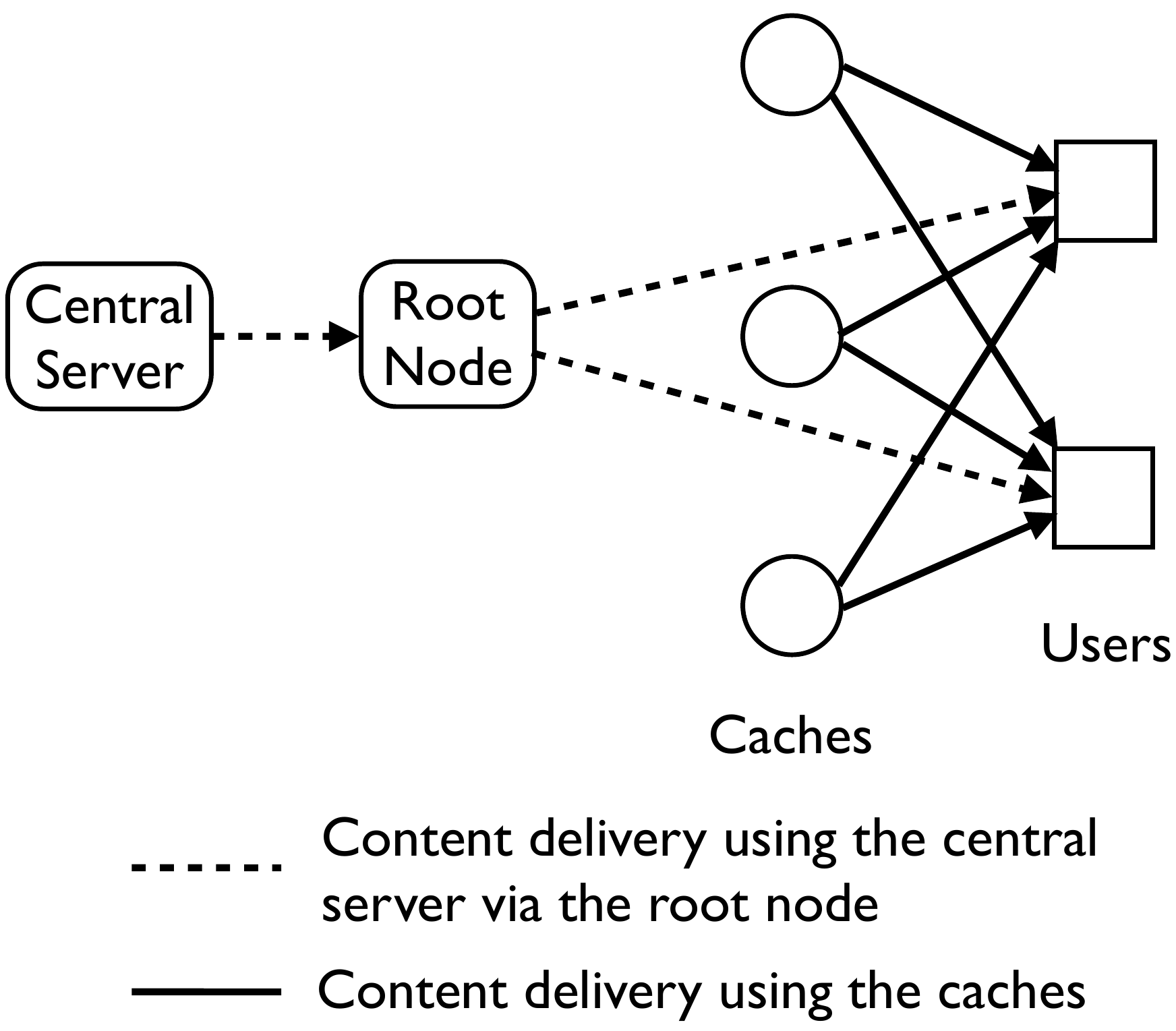}
		\caption{\sl An illustration of a cache cluster consisting of three caches serving two users. Each user can either be served by the caches or by the central server via the root node. \label{fig:cache_cluster}}
	\end{center}
\end{figure}

Recent works on content replication strategies in content delivery systems focus on the setting where each user request is served by only one cache and each cache serves only one request at a time \cite{SGSS14, LLM12,LLM13, Whitt07, XT13,maddah2014fundamental, maddah2013decentralized, pedarsani2014online, niesen2014coded, hachem2014multi, zhang2015coded, shanmugam2013femtocaching, borst2010distributed, moharir2017content}. In a departure from these works, we explore the benefits of relaxing this constraint by allowing caches to pool their resources, $i.e.$, allowing a request to be served by multiple caches. \textit{We refer to this phenomenon of multiple caches pooling their resources to serve a request as ``resource pooling".} (Simply,) in resource pooling, different parts of the requested file can be delivered to the user by different caches. While resource pooling can enhance performance, it comes at the cost of an increase in coordination overheads, thus motivating the need to achieve the desired performance with as little resource pooling as possible. To characterize the benefits of limited resource pooling, we limit the number of requests each cache can serve concurrently and evaluate the performance of the system as a function of this limit. 

The key takeaway of this work is that the benefits of resource pooling vary drastically with the popularity profile of contents. More specifically, we show that when popularity is comparable across contents, even a small amount of resource pooling leads to a huge improvement in performance. In contrast, if content popularity is lopsided, the benefits of resource pooling are very limited. 

\section{Related Work}
\label{section:related_works}
Content caching has a rich and varied history, see for example~\cite{wessels2001}
and references therein. More recently, it has been studied in the
context of video-on-demand systems for which efficient content placement
schemes have been proposed in~\cite{borst2010distributed, tan2013optimal} among others. See \cite{paschos2018role} for a recent overview of various challenges in content caching and delivery networks, and a summary of known results. Due to limited space, we mention here only those works which are closest to our setting.


Motivated by studies like \cite{krishnan2013video}, which observe that users of VoD services are delay intolerant, we focus on the setting where requests are never queued and each request is served immediately, either by the caches or the central server. In the studies of \cite{SGSS14, LLM12, LLM13, Whitt07, XT13} the focus is on the setting where each request can be served by any one cache and the central server communicates with each user separately. The focus in the studies of \cite{LLM12, LLM13, Whitt07} is on the setting where content popularity is known, whereas the studies of \cite{SGSS14, XT13} focus on the setting where content popularity is unknown. In the optimal caching policies proposed in these works, the number of caches storing a file is a non-decreasing function of file popularity. In this work, we see that this is not necessarily true for our setting where requests arrive in a batch and the central server communicates with the multiple users simultaneously via an error free broadcast link.


The setting where each user is pre-matched to a cache and the central server communicates with the users via an error free broadcast link has been studied recently in \cite{maddah2014fundamental, maddah2013decentralized, pedarsani2014online, niesen2014coded, hachem2014multi, zhang2015coded}. The key intuition derived in these studies is that content should be placed in the caches not only to provide local access but to also help generate coded-multicasting opportunities in the delivery phase which can serve multiple user requests simultaneously. It is also shown that exploiting such coding opportunities in the delivery phase is necessary for optimal performance.

Our setting differs from the two settings discussed above as each request can be served by multiple caches (resource pooling) and the central server communicates with the multiple users simultaneously via an error free broadcast link.  Unlike our setting,  the studies in \cite{shah2014performance, shah2015high} characterize the benefits of resource pooling in the setting where jobs are allowed to be queued at the caches. In addition, the studies of \cite{shah2014performance, shah2015high} focus on the case where all contents are equally popular, whereas we allow for more general popularity profiles. 
In a preliminary version of this work, we showed that, for this setting, coding caching is not always necessary for optimal performance \cite{moharir2017content}. Refer to \cite{moharir2017content} for a detailed discussion on the difference in nature of the optimal caching policies for the settings discussed thus far. In this work, we generalize the setting studied in \cite{moharir2017content} by allowing requests to be served by multiple caches and letting each a cache serve multiple users simultaneously. In \cite{reddy2018effects}, we also characterize the effect of storage heterogeneity in distributed cache systems.

The rest of the paper is organized as follows. Section \ref{sec:setting} briefly describes our problem setting.  Sections \ref{section:notations} and
\ref{section:preliminaries} describe some useful notations and preliminaries. Sections \ref{sec:results} and \ref{sec:simulations} describe our theoretical  results and simulation results. Section \ref{sec:conclusions} summarizes our paper and Section \ref{sec:proof} gives the proofs of our results mentioned in Section \ref{sec:results}.


	\section{setting} \label{sec:setting}

We study a system consisting of a central server, and multiple caches with limited storage as well as limited service capabilities. The system offers a content catalog consisting of $n$ contents\footnote{Throughout the paper, we will use `content' and `file' interchangeably to denote individual elements of the catalog.} of equal size (say $1$ unit = $b$ bits), where the number of contents ($n$) and the number of caches ($m$) are of  the same order ($i.e.$,  $n=mc$, for some constant $c>0$). Users make requests for various contents from the catalog, which have to be served using the caches and the central server. 

The system operates in two phases: the \textit{placement phase} and the \textit{delivery phase}.  During the placement phase, each cache stores content related to the $n$ files in the catalog. After the placement phase has concluded, the system moves to the delivery phase in which a batch of requests arrives and has to be allocated to the caches for service. While we allow the splitting of files into parts, unlike \cite{maddah2014fundamental, zhang2015coded}, we restrict our attention to \textit{uncoded policies, which do not employ any coding in either placement phase or delivery phase.}

We are interested in the asymptotic performance of this system as $n$, $m$ $\rightarrow$ $\infty$.
\subsection{Storage Model} \label{storage}
The central server stores the entire catalog of $n$ contents offered by the content delivery system and each of the $m$ caches has the capacity to store $k$ units of data. As mentioned before, we allow files to be split into smaller parts and caches to store a subset of the parts of any file.
\subsection{Request Model} \label{request}
In each time-slot, requests arrive in batches of size $r =\rho m$, for some constant $0<\rho<1$. Each request is generated according to an independent and identically distributed process where the probability of the requested content being Content $i$ be denoted by $p_i$. We analyze the performance of the system when the $p_i$'s satisfy the Zipf distribution which is defined as follows: the fraction of requests for the $i^{\text{th}}$ most popular content is proportional to $i^{-\beta}$, where $\beta\geq 0$ is a constant, known as the Zipf parameter. This choice is motivated by the fact that empirical studies of many VoD services have shown that the content popularity distributions match well with the Zipf distribution \cite{liu2013measurement,BC99,YZ06,fricker2012impact}. As the value of $\beta$ increases, the content popularity profile becomes more lopsided. Typical values of  $\beta$ lie between 0.6 and 2 (\cite{liu2013measurement,BC99,YZ06,fricker2012impact}). 
%
%
\subsection{Service Model} \label{service}
All the user requests have to be served jointly by the caches and the central server. Every user request is assigned to one or more caches, each of which uses its stored content to provide various parts of the requested file.  Due to hardware, power and/or bandwidth constraints, the user-cache assignment needs to satisfy two restrictions: each cache can only serve up to $a$ requests  and the total data served by a cache should not exceed $1$ unit. There is no restriction on the number of caches that serve a particular request. 

The root node can also enlist the help of the central server to assist with serving the user requests. Some requests are served by the caches. To serve the remaining requests, the central server transmits the requested files or parts thereof to the root node, which then forwards them to the users.   See Figure \ref{fig:cache_cluster} for an illustration. Using the data received from the assigned caches and the central server, each user should be able to reconstruct its requested file. Refer to Figure \ref{fig:example} for an example. 

\subsection{Goal} \label{goal}
The reason for deploying local caches is that they can help reduce the communication on the bottleneck link between the central server and the root node. Our goal in this paper is to design placement and delivery schemes which minimize the expected transmission rate of the central server needed to satisfy all the user requests, where the expectation is with respect to the popularity distribution of the user requests. Note that if a file needs to be sent by the central server via the root node to more than one users in a batch, the central server transmits it to the root node only once. In order to achieve this objective, we utilize the knowledge of the content popularity profile to design appropriate storage and service policies. 

\begin{figure}[t]
	\begin{center}
		\includegraphics[scale=0.33]{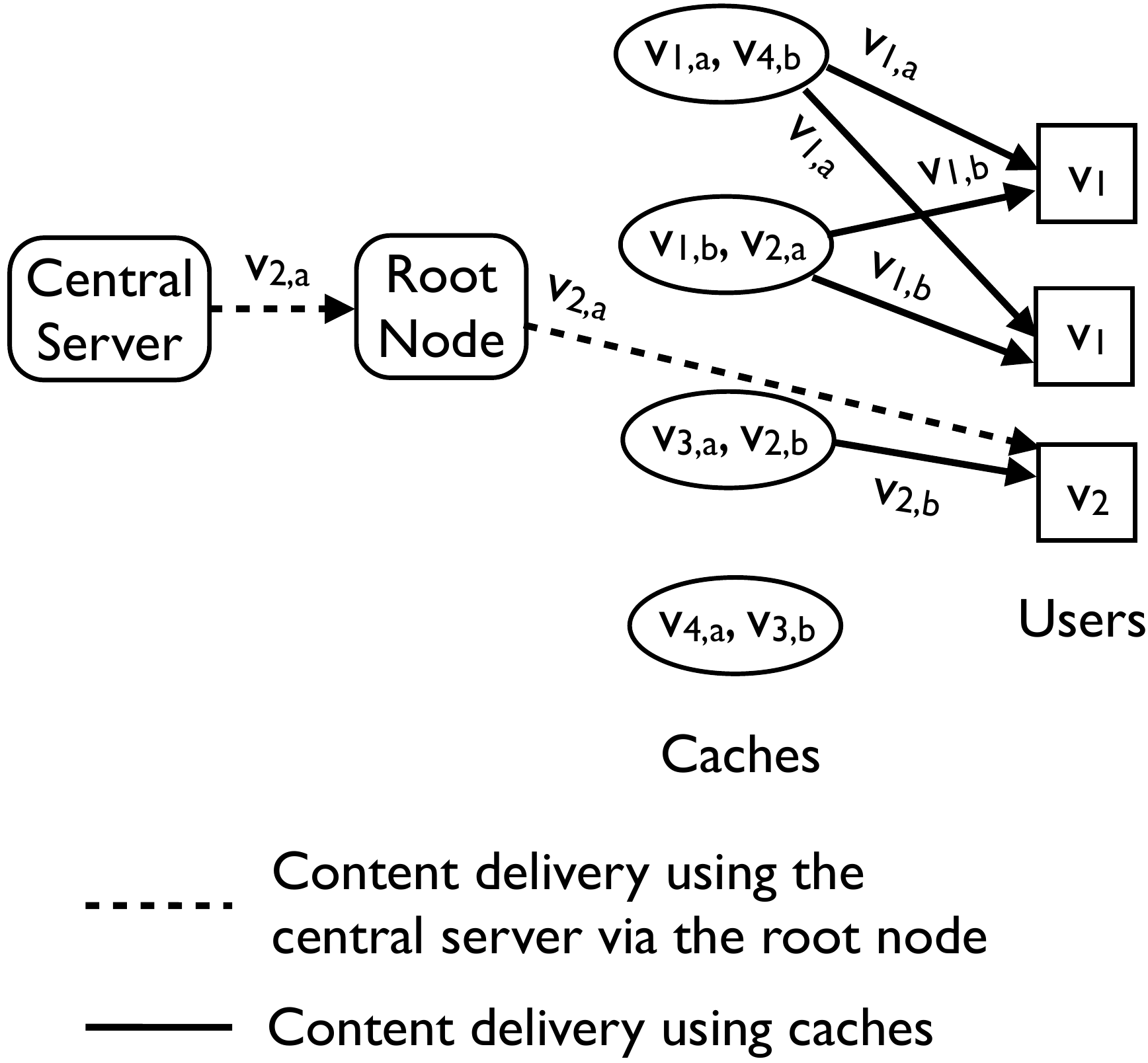}
		\caption{\sl An illustration of a cache cluster consisting of four caches serving three users. The catalog consists of 4 files, $v_i$, $i=\{1,2,3,4\}$, each of which is divided into two equal parts ($v_{i,a}$, $v_{i,b}$) to store on the caches. Each cache can serve upto two requests ($a=2$) as long as the total data delivered by each cache is not more than 1 file. The first two users request file $v_1$ and are served by the first two caches. The third user requests file $v_2$ and receives the first part of the file from the central server and the second part from the third cache. The third user cannot be served by the second cache even though it stores the requested file as that will violate the total data output constraint of that cache. \label{fig:example}}
	\end{center}
\end{figure}

 {
\subsection{Contributions}
The main goal of this work is to analyze the impact of resource pooling on the performance of the content caching and delivery system described above. In particular, we propose efficient placement and delivery schemes for our setting and characterize the variation of the central server transmission rate with the resource pooling parameter $a$ which denotes the number of requests that each cache can serve  simultaneously.  Recall our assumption that the requests follow a Zipf distribution with parameter $\beta$; we find that the impact of resource pooling on the server transmission rate is qualitatively very different for the cases of $0\leq\beta < 1$ and $1<\beta<2$. The former corresponds to the case where the popularity is comparable across contents, whereas the latter represents a scenario where the content popularity profile is lopsided with a few very popular contents. We deal with these two cases separately.
%

\textit{$0\leq\beta < 1$}: We extend the proportional placement and optimal matching delivery scheme proposed in \cite{LLM12} for $a=1$ to the case of resource pooling with $a > 1$. The scheme splits each file into $a$ equal-sized sub-files, creates copies of each sub-file in proportion to its popularity and then stores them across the caches in the system. The delivery procedure splits each file request into $a$ sub-requests, one for each of its $a$ sub-files and then matches as many sub-requests as possible to caches hosting the corresponding sub-files, while ensuring that no cache is assigned to more than $a$ requests. The unmatched sub-requests are served directly by the central server. Theorem~\ref{thm:zipf01} presents an upper bound on the expected server transmission rate of the proposed scheme as well as a lower bound on the performance of any uncoded policy for this setting. In particular, we show that the expected server transmission rate of the proposed scheme decays exponentially with $ak$, i.e., the product of the resource pooling parameter $a$ and the storage capacity per cache $k$. Thus, if the popularity follows the Zipf distribution with parameter $0\leq\beta < 1$, a small amount of resource pooling can lead to a significant reduction in the transmission rate of the central server. As a corollary of this result, we find that for $k \ge 1$ and each cache serving a number of requests growing only logarithmically in the total number of requests, i.e., $a = \Omega(\ln n)$, all the user requests can be served with just a vanishing expected server transmission rate. 
%

\textit{$1<\beta<2$}: For this setting, our proposed placement scheme is based on the solution to an appropriate fractional Knapsack problem \cite{goodrich2006algorithm} which specifies the number of copies of each content to be stored across the various caches. The delivery scheme matches requests to appropriate caches one by one, starting from the least popular files. As before, the unmatched requests are served directly by the central server. Theorem~\ref{thm:zipf12u} provides an upper bound on the expected server transmission rate for our scheme. Comparing this to the lower bound for any uncoded policy presented in Theorem~\ref{thm:zipf12}, we conclude that if content popularity follows the Zipf's distribution with $1<\beta<2$, our Knapsack Storage + Match Least Popular policy (KS+MLP) is order-optimal. Furthermore, as a corollary, we find that in most cache memory regimes, the order-optimal rate can be achieved with $a = 1$, i.e., without any resource pooling. Thus, if the popularity follows the Zipf distribution with parameter $\beta \in (1,2)$, there is very limited order-wise benefit of resource pooling ($a > 1$). This is in contrast to the case of $0\leq\beta < 1$. 

Note that the assumptions made in this paper are commonly used in the existing literature for technical simplicity. Please see for example \cite{LLM12,LLM13}. These papers consider a similar setting without resource
	pooling (i.e., $a = 1$ case). The policies proposed in our paper can be used for more general settings. In particular, (\emph{i}) the assumption that the all files are equal size is often made for analytical tractability and a possible solution for unequal file sizes is splitting files into segments of equal size, (\emph{ii}) the assumption that the  number of caches and number of contents have same order makes our analysis simple. Our analysis can be
	easily extended to the setting where the number of contents is of a higher order. 
	In \cite{reddy2018effects}, we relaxed this assumption for the setting without resource pooling, (\emph{iii}) the homogeneous assumptions on the caches (i.e., storage and service capabilities of all caches are equal) is also to make our analysis simple. Our policies also work for heterogeneous cache sizes. In \cite{reddy2018effects}, we study the effect of storage heterogeneity  for the setting without resource pooling. 
	
	Due to practical limitations, like hardware, power and/or bandwidth constraints, there exists a limit on the amount of data that can be sent from a cache to the users in each time-slot. We have therefore imposed an upper limit on the output data rate of each cache and normalized it to $1$ unit per time-slot. Serving a request via multiple caches increases the synchronization overheads. Hence, we have imposed an upper limit on the number of caches used to serve a request as $a$. Studying the effect of heterogeneous service capabilities and the effect of coded policies is a promising direction of study and is beyond the scope of this work.

	\section{Notations}
\label{section:notations}
We use the notations mentioned in Table \ref{notati} and definitions in Table \ref{definat} in the rest of this paper. 
	\begin{table}[h]
				\begin{center}
			\begin{tabular}{|c|c|}
				\hline
				\textbf{Symbol} & \textbf{Meaning} \\
				\hline
					\hline
				$m$&number of caches\\
				\hline
				$n$&number of files\\ 
				\hline
				$r$ & number of requests \\ 
				\hline
				$k$ & storage capacity of each cache\\
				\hline 
				$a$ & maximum number of requests \\ & each cache can serve\\
				\hline
				$\beta$ & Zipf parameter \\
				\hline 
				$b$ & File size (1 unit = $b$ bits)\\
				\hline
				$p_i$ & the request probability for File  $i$ \\
				\hline
			\end{tabular}
			\vspace{0.1cm}
			\caption{Notations}\label{notati}
				\end{center}
			\end{table}
		\begin{table}[h]
			\begin{tabular}{|c|c|}
				\hline
				\textbf{Notation} & \textbf{Definition}\\
				\hline
				\hline
				$f(n)=\oo\big(g(n)\big)$&  $\lim\limits_{n\rightarrow\infty}\frac{f(n)}{g(n)}\rightarrow 0$ \\
				\hline 
				$f(n)=\OO\big(g(n)\big)$&$\exists C \text{ s.t. } \lim\limits_{n\rightarrow\infty}\frac{f(n)}{g(n)}\leq C$ \\
				\hline 
				$f(n)=\Theta\big(g(n)\big)$&$ \exists C_1, C_2 \text{ s.t. } C_1\leq \lim\limits_{n\rightarrow\infty}\frac{f(n)}{g(n)}\leq C_2$\\
				\hline 
				$f(n)=\Omega\big(g(n)\big)$&$\exists C \text{ s.t. } \lim\limits_{n\rightarrow\infty}\frac{f(n)}{g(n)}\geq C $\\
				\hline 
				$f(n)=\omega\big(g(n)\big)$&$  \lim\limits_{n\rightarrow\infty}\frac{f(n)}{g(n)}\rightarrow \infty$\\
				\hline  
			\end{tabular}
		\vspace{0.1cm}
		\caption{Definitions}\label{definat}
	\end{table}
%
%

\section{Preliminaries}
\label{section:preliminaries}
Our proposed scheme is based on the solution to the fractional Knapsack problem \cite{goodrich2006algorithm}, which can informally be defined as follows: choose items to keep in the knapsack such that the cumulative value of the items is maximized, while ensuring that the cumulative weight of the items is not more than the knapsack's capacity. Formally, if the total capacity of the knapsack is $W$, item $j$ has value $v_j$ and weight $w_j$, the fractional knapsack problem is defined as:
\begin{eqnarray*}
	\max_{\{x_j\}} \displaystyle \sum_{j=1}^J x_j v_j \hspace{0.23in} \text{ s.t. } \displaystyle \sum_{j=1}^J x_j w_j \leq W, \hspace{0.125in} \& \hspace{0.125in} 0 \leq x_j \leq 1, \text{ } \forall j.
\end{eqnarray*}

\noindent Without loss of generality, let the items be indexed in decreasing order of value to weight ratio, i.e., 
 $
\frac{v_1}{w_1} \geq \frac{v_2}{w_2} \geq ... \geq \frac{v_J}{w_J}
$.

\noindent Let $j^*$ be such that
	$\sum_{j=1}^{j^*-1} w_j \leq W, \text{ and } \sum_{j=1}^{j^*} w_j > W.$ 
The solution to the fractional Knapsack problem is:
\begin{eqnarray*}
	x_j = \begin{cases}
		1, & \text{for } j < j^*,  \\
		\dfrac{W - \sum_{j=1}^{j^*-1} w_j}{w_{j^*}}, & \text{for } j = j^*,\\
		0 & \text{otherwise}.
	\end{cases}
\end{eqnarray*}
\begin{remark}
	The time complexity of the fractional Knapsack problem solution is  $\OO(J \log J)$. 
\end{remark}


\section{Main results and discussion} \label{sec:results}
In this section we state and discuss our main results. We relegate the proofs to Section \ref{sec:proof}. 

\subsection{Zipf distribution with $\beta \in [0,1)$}
\label{zipf} 
We first state our results for the case where content popularity follows the Zipf distribution (defined in Section \ref{sec:setting}) with parameter $\beta \in [0,1)$.  We propose a storage/service policy for this setting and evaluate its performance.

Our storage policy is inspired by the  Proportional Placement (PP) policy proposed in \cite{LLM12}. We divide each file into $a$ sub-files of equal size. Note that the popularity of a sub-file is the same as the popularity of the corresponding file as a whole. The number of caches storing each sub-file  is proportional to its popularity.  We ensure that no cache stores more than one sub-file of the same file.  
 
Our service policy is as follows: we treat each request for a file as $a$ sub-requests, one for each of the $a$ sub-files and create a bipartite graph $G(V_1,V_2, E)$, where $V_1$ is the set of sub-requests, $V_2$ is the set of caches, and $E$ is the set of edges. There is an edge between  $v_1\in V_1$ and $v_2 \in V_2$ if Cache $v_2$ can serve Sub-request $v_1$, $i.e.$, if it stores a copy of the requested content. We construct a new set of nodes $V_2^{(a)}$ which contains $a$ copies of each node in $V_2$ and find the maximum cardinality matching between the set of requests ($V_1$) and the set of caches ($V_2^{(a)}$).
All the sub-requests matched to a copy of $v_2 \in V_2^{(a)}$ are served by Cache $v_2$ and all the sub-requests which are not matched to any cache are served by the central server via the root node. We refer to this service policy as the Optimal Matching Routing (OMR) policy. Note that this policy satisfies our service constraints on caches: ($\emph{i}$) each cache can only serve up to $a$ requests  and ($\emph{ii}$) the total data served by a cache should not exceed $1$ unit. We refer to this scheme as the Proportional Placement + Optimal Matching Routing (PP+OMR) policy. 
\begin{theorem} \label{thm:zipf01}
Consider a system with $n$ files with popularity following the Zipf distribution with parameter $\beta \in [0,1)$, and $m=n/c$ ($c>0$, $c$ is a constant) caches of size $k $ units each. Every cache can serve at most $a$ requests and the total data served by a cache cannot exceed $1$ unit. The system receives a batch of  $r=\rho m $ ($0<\rho<1$, $\rho$ is a constant) i.i.d. requests.\\
(a) Let $R_{z_1}$ be the central server's transmission rate for our 
 \indent \hspace{0.5ex} policy described above. Then, 
	\begin{align*}
	\mathbb{E}[R_{z_1}]
	= \begin{cases}
	\OO(n)& \text{if } k < c, \text{ } \forall a, \\
	\OO\big( \min\{n,nk\exp (-c_1ak)\}\big) &\text{if } k \geq c, 
	\end{cases}
	\end{align*}
	\hspace{3ex} where $c_1$ is a constant\footnote{independent of $n$, $a$, $k$}, and is greater than zero. \\
	(b) Let $R^*_{z_1}$ be the central server's transmission rate for the  \indent \hspace{0.5ex} optimal  uncoded policy. Then, 
	\begin{eqnarray*}
		\mathbb{E}[R^*_{z_1}]
		= \begin{cases}
			\Omega(n)&\text{if } k < c - \Theta(1), \text{ } \forall a, \\
			\Omega\big(n\exp({-c_2ak\ln{ak}})\big) &\text{if } k \geq c,
		\end{cases}
	\end{eqnarray*}
\hspace{3ex}  where $c_2$ is a constant  and is greater than zero.
\end{theorem}

From the first part of the theorem we conclude that the performance of our policy depends on the product $ak$, where $k$ is the number of files each cache can store and $a$ is the number of requests each cache can serve simultaneously. As expected, the performance of our policy improves with increasing cache memory. In addition, for a fixed amount of cache memory, the performance of our policy can be improved by increasing resource pooling. The second part of the theorem gives a lower bound on the expected transmission rate of the central server under any uncoded storage/service policy which satisfies the assumptions in Section \ref{sec:setting}.

\begin{corollary}\label{cor:zipf_0to1}
	Consider a system with $n$ files with popularity following the Zipf distribution with parameter $\beta \in [0,1)$, and $m=n/c$ ($c>0$, $c$ is a constant) caches of size $k(\geq c) $ units each. Every cache is capable of  serving at most $a$ requests and the total data served by a cache cannot exceed $1$ unit. The system receives a batch of  $r=\rho m $ ($0<\rho<1$, $\rho$ is a constant) i.i.d. requests.
	\begin{enumerate}
		\item If $ak=\OO((\ln n)^{\alpha})$ for $\alpha<1$,  $\mathbb{E}[R_{z_1}^*]=\omega(1)$.
		\item If $ak = \Omega\big(\ln (n)\big)$, then for our storage/service policy, $\mathbb{E}[R_{z_1}]=\oo(1)$. 
	\end{enumerate}
\end{corollary}

We conclude that for $ak = \OO((\ln n)^{\alpha})$ for $\alpha<1$, no uncoded storage/service policy can bring the expected transmission rate of the central server down to zero.  In addition, for $k \geq c$, $ak = \Omega(\ln n)$ is sufficient to ensure that with high probability, all requests are served by the caches under our storage/service policy. We thus conclude that if the popularity follows the Zipf distribution with parameter $\beta \in [0,1)$, a small amount of resource pooling can lead to a significant reduction in the transmission rate of the central server.


\subsection{Zipf distribution with $\beta \in (1,2)$}\label{zipfb}

We now focus on the case where content popularity follows the Zipf distribution with parameter $\beta \in (1,2)$.  The following proposition gives a lower bound on the expected transmission rate from the central server for a slightly less restricted system than the one mentioned in Section \ref{sec:setting}. Hence, Proposition \ref{prop:zipf12} provides a lower bound for our system as well. 

\begin{proposition}
	\label{prop:zipf12}
	Consider a distributed cache system with $n$ contents { each of size $b_i$ bits}, $m$ caches of size $k$ units each, and a batch of $r$ requests arriving at the beginning of each time-slot. Each request is generated according to an i.i.d. process, and the request probability for Content $i$ is denoted by $p_i$. Let $R^*_{\text{NC}}$ denotes the minimum transmission rate required to serve all requests arriving in a batch using uncoded storage and service policies, under the constraint that any cache can serve upto $a$ requests for each of its stored content\footnote{In our system, the total number of served requests across all stored content in a cache is at most $a$}. Then, we have that,
	\begin{eqnarray*}
		\EE[R^*_{\text{NC}}] &=& 
		\Omega\bigg(\displaystyle \sum_{i=1}^n b_i (1-(1-p_i)^{r}) - \text{O}^*\bigg), \\
		\text{where, O}^* &=& \max_{\{x_{i,u}\}} \displaystyle \sum_{i=1}^n \displaystyle \sum_{u=1}^{b_i} x_{i,u} (1-(1-p_i)^{r}) \\
		&& \text{s.t. }  \displaystyle \sum_{i=1}^n \displaystyle \sum_{u=1}^{b_i} x_{i,u} \max\bigg\{\Big\lfloor\frac{rp_i}{a}\Big\rfloor,1 \bigg\}  \leq mkb, \\
		&& 0 \leq x_{i,u} \leq 1, \text{ } \forall i,u.
	\end{eqnarray*}

\end{proposition}

%

The quantity $\text{O}^*$ defined in Proposition \ref{prop:zipf12} is the solution to the fractional Knapsack problem described in Section \ref{section:preliminaries} with:

\begin{itemize}
	\item[--] The value of Bit $u$ of Content $i$, $$v_{i,u}=1-(1-p_i)^{r},$$  is the probability that Content $i$ is requested at least once. 
	\item[--] The weight of Bit $u$ of Content $i$, $$w_{i,u} = \max\bigg\{\Big\lfloor\frac{rp_i}{a}\Big\rfloor,1 \bigg\} ,$$ where $rp_i$ is the expected number of requests for Content $i$ in a time-slot. 
	\item[--] The capacity of the knapsack, $$W = mkb,$$ is the total memory of the $m$ caches. 
\end{itemize}
$x_{i,u}=1$ implies that $\max\Big\{\Big\lfloor\frac{rp_i}{a}\Big\rfloor,1 \Big\} $ copies of Bit $u$ of Content $i$ are stored in the knapsack, and, $x_{i,u}=0$ implies that Bit $u$ of Content $i$ is not stored in the knapsack. Proposition \ref{prop:zipf12} lower bounds the expected transmission rate by 
$$\displaystyle \sum_{i=1}^n \displaystyle \sum_{u=1}^{b_i} (1-x_{i,u})(1-(1-p_i)^{r}),$$
which is the expected number of files not stored in knapsack and requested at least once.

Next, we evaluate a lower bound for the Zipf distribution with parameter $\beta \in (1,2)$. Our system mentioned in Section \ref{sec:setting} is more restricted than the system mentioned in Proposition \ref{prop:zipf12} and its content popularity follows the i.i.d. Zipf distribution. Hence, replacing the $p_i$'s accordingly in Proposition \ref{prop:zipf12} gives a valid lower bound for our system.

Recall that the solution to the fractional Knapsack problem (Section \ref{section:preliminaries}) is obtained by ranking the items in decreasing order of the value to weight ratio and choosing the maximum number of highest ranked items such that their cumulative weight is less than the knapsack capacity.

For the  Zipf distribution with parameter $\beta$, let $\tilde{i} = \Big\lceil(\frac{rp_1}{2a})^{1/\beta}\Big\rceil$ and let $z_i$ be the value to weight ratio  of Content $i$ (for all $u$). We have that, for $p_i=\frac{p_1}{i^\beta}$,
\begin{eqnarray*}
	z_i = \dfrac{v_{i,u}}{w_{i,u}} = \begin{cases}
		\dfrac{1-(1-p_i)^{r}}{\big\lfloor\frac{rp_i}{a}\big\rfloor}, & \text{ for } i \leq \tilde{i}, \\
		1-(1-p_i)^{r}, & \text{   for } i > \tilde{i},
	\end{cases}
\end{eqnarray*}

\begin{figure}[t]
	\begin{center}
		\includegraphics[width=3.25 in]{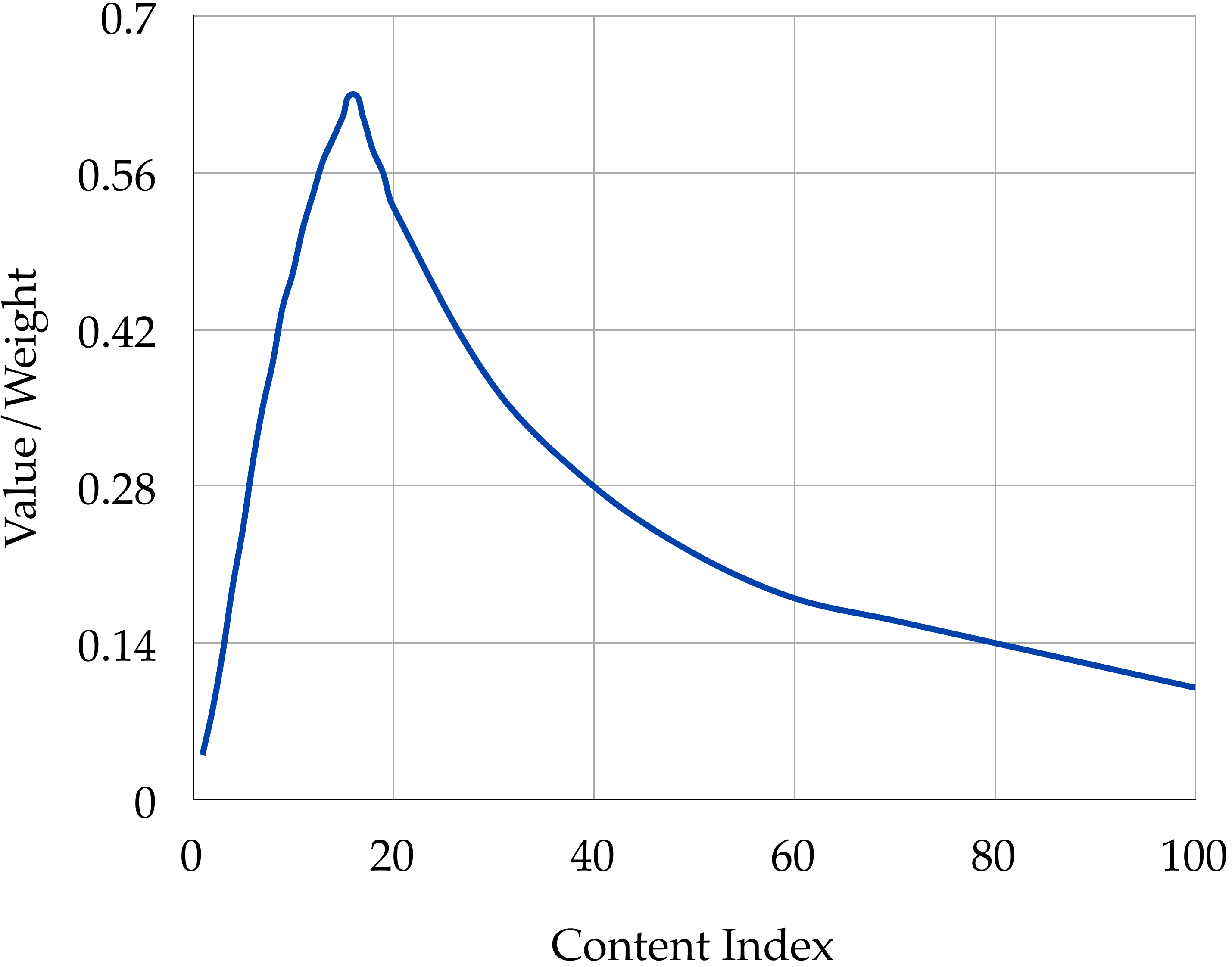}
		\caption{Value to weight ratio when content popularity follows the Zipf distribution for $n=m=r =100$, and $\beta = 1.2$.}\label{Zipf_ratio}
	\end{center}
\end{figure}

Given this, $z_i$ increases from $i=1$ to $\tilde{i}$ and decreases from $i=\tilde{i}+1$ to $n$. For example, Figure \ref{Zipf_ratio} illustrates how the ratio of the value to weight ratio for $n=m={r} =100$, $a=1$ and $\beta = 1.2$ varies as a function of content index. 

Hence, the optimal solution has the following structure: $\exists$ $i_{\text{min}}$, $i_{\text{max}}$ with $i_{\text{min}} \leq \tilde{i} \leq i_{\text{max}}$, such that, for $0\leq f_1, f_2 \leq 1$,
\begin{eqnarray*}
	x_{i,u} = \begin{cases}
		1, & \text{  if } i_{\text{min}} < {{i}} < i_{\text{max}}, \\
		f_1, & \text{if } i =i_{\text{min}} , \\
		f_2, & \text{if } i =i_{\text{max}} , \\
		0, & \text{otherwise},
	\end{cases}
\end{eqnarray*}

We optimize over $i_{\text{min}}$ and $i_{\text{max}}$ to get a lower bound on the expected transmission rate for particular values of $c$, $k$, {$\rho$} and $m$.  Theorem \ref{thm:zipf12} shows the results for the case where content popularity follows the Zipf distribution with parameter $\beta$, such that $1 < \beta <2$. 

\begin{theorem}\label{thm:zipf12}{\em [Lower bound]}
Consider a system with $n$ files, each of size $1$ unit with popularity following the Zipf distribution with parameter $\beta \in (1,2)$,  and $m=n/c$ ($c>0$, $c$ is a constant) caches of size $k $ units  each. Every cache is capable of  serving at most $a$ requests  and the total data served by a cache cannot exceed $1$ unit. The system receives a batch of  $r=\rho m $ ($0<\rho<1$, $\rho$ is a constant) i.i.d. requests. Let $R^*_{z_2}$ be the central server's transmission rate for the optimal policy. Then, 
\begin{eqnarray*}
 \mathbb{E}[R^*_{z_2}]
  = \begin{cases}
            \Omega(n^{2-\beta})&\text{if } k <  c - \Theta(1), \text{ } \forall a, \\
           \Omega\left(n^{(2-\beta-\gamma) / \beta }\right) &\text{if } k = c, \ a=m^{\gamma}, \gamma \in[0,1],\\
           \Omega(0)&\text{if } k>c+\Theta(1), \text{ } \forall a.
         \end{cases}
\end{eqnarray*}
\end{theorem}
%

\begin{remark} 
	\label{remark:insights}
For contents expected to be requested at least once, i.e., Contents $i$ such that $\frac{rp_i}{a} \geq 1$, it is optimal to store contents with lower popularity. Intuitively, given that two contents are going to be requested at least once each, all the requests for the less popular content can be served using fewer caches and a lesser amount of storage than the more popular content. Therefore, between the two contents, storing the less popular content reduces the transmission rate by $1$ unit using fewer memory resources. 
		
		For contents expected to be requested at most once, i.e., Contents $i$ such that $\frac{rp_i}{a} < 1$, it is optimal to store the more popular contents. Intuitively, between two contents with the same weight, storing the more popular content increases the probability of reducing the transmission rate required to serve incoming requests, while using the same amount of memory resources.
		
		Amongst the contents for each of which adequate number of copies have been stored, it is optimal to serve the one with fewer  requests. Intuitively, between two such contents, serving the content with less requests reduces the transmission rate by $1$ unit using fewer resources.
\end{remark}

Inspired by the above insights, we propose a storage and service policy, whose transmission rate is order-wise equal to the  lower bound in Theorem \ref{thm:zipf12}. We refer to this policy as Knapsack Storage + Match Least Popular (KS+MLP) policy.  

Our storage policy is inspired by the  Knapsack Storage policy, and is described in two parts. \\

\noindent \textbf{Knapsack Storage: Part 1} -- The first part of the Knapsack Storage policy determines how many caches each content is stored on by solving a fractional Knapsack problem \cite{goodrich2006algorithm}. The parameters of the fractional Knapsack problem are as follows:
{
\begin{itemize}
	\item[--] An unstored content will be broadcasted if it is requested at least once. Hence, the value of  Content $i$, 
	\begin{equation}\label{eq:value1}
	v_{i}=1-(1-p_i)^{r},
	\end{equation} is the probability that  Content $i$ is requested at least once in the time-slot. 	
	\item[--] The weight of Content $i$ ($w_{i}$) represents the number of caches  Content $i$ will be stored on if selected by the Knapsack problem. If we decide to store a content on the caches, we would like to ensure that all requests for that content can be served by the caches, so that the content need not to be transmitted by the central server. To ensure this, we fix $w_{i}$ to be high enough to ensure that with high probability, i.e., with probability $\rightarrow 1$ as $\tilde{m},m,n \rightarrow \infty$, the number of requests for  Content $i$ in a time-slot is less than or equal to $w_{i}$. We use the following values for the $w_{i}$'s:{
	\begin{eqnarray}\label{eq:weight1}
	w_{i} = \begin{cases}
	\frac{m}{a}, & \text{if } i = 1 \\
	\big\lceil \big(1 + \frac{p_1}{2}\big) \frac{rp_i}{a} \big\rceil, & \text{if } 1 < i \leq n_1 , \\
	\Big\lceil  \frac{4 p_1(\log m)^2}{a} \Big\rceil, & \text{if } n_1 < i \leq n_2, \\
	\big\lceil\frac{4}{a\delta}\big\rceil, & \text{if } n_2 < i \leq n. 
	\end{cases} 
	\end{eqnarray}
	where $n_1=\frac{(rp_1)^\frac{1}{\beta}}{(\log m)^\frac{2}{\beta}}$, and $n_2=m^{\frac{1+\delta}{\beta}}$ for some $0<\delta<\beta-1$.}
\end{itemize}

Since the value of Content $i$, $v_{i}=1-(1-p_i)^{r},$ is the probability that Content $i$ is requested at least once in the time-slot, {$ \sum_{i=1}^n(1-x_{i}) v_{i},$}
is the expected number of contents that are not stored in the knapsack and are requested at least once. As a result, maximizing $\sum_{i=1}^n x_{i} v_{i}$} minimizes the expected number of contents that are not stored in the knapsack and are requested at least once, which is equivalent to minimizing the expected transmission rate.

Figure \ref{fig:knapsack_storage_Part 1} formally describes Knapsack Storage: Part 1. 
\begin{figure}[h]
	\hrule
	\vspace{0.1in}
	\begin{algorithmic}[1]
		\STATE Solve the following fractional Knapsack problem
		\begin{eqnarray*}
			\max_{\{x_i\}} &&  \displaystyle \sum_{i=1}^n x_i v_i \\
			\text{s.t. }  &&b \displaystyle \sum_{i=1}^n x_i w_i  \leq mk, \\
			&& 0 \leq x_i \leq 1, \text{ } \forall i,
		\end{eqnarray*}
		where $v_i$s and $w_i$s are as defined in equations \ref{eq:value1} and \ref{eq:weight1} respectively.
		\STATE The set of contents to be stored $$S = \{\lfloor x_i \rfloor w_i \text{ copies of Content } i, 1 \leq i \leq n\}.$$
	\end{algorithmic}
	\vspace{0.1in}
	\hrule
	\caption{Knapsack Storage: Part 1 -- \sl Determines how many caches each content is stored on.}
	\label{fig:knapsack_storage_Part 1}
\end{figure}

\begin{remark}
	{{Recall from Remark \ref{remark:insights}} that the optimal solution to the fractional Knapsack problem prioritizes selecting contents with larger value to weight ratios. Therefore, for certain values of the system parameters ($n$, $m$, $r$, $k$), the optimal solution to the fractional Knapsack problem in Figure \ref{fig:knapsack_storage_Part 1} does not store the most popular contents on the caches. {As discussed in Remark \ref{remark:insights}, intuitively, in order to serve all the requests for a popular content via the caches, the content needs to be replicated on a large number of caches, since each cache can only serve $a$ requests at a time. It follows that, at times, it is better to serve all the requests for a popular content via a single transmission from the central server, instead of replicating it on a large number of caches, thus using up a lot of memory resources.}}  
	
\end{remark}

\textbf{Knapsack Storage: Part 2} -- The next decision to be made is which contents to store on which caches, i.e., how to partition the set of contents selected by Knapsack Storage: Part 1 (Figure \ref{fig:knapsack_storage_Part 1}) into $m$ groups.

\begin{figure}[h]
	\hrule
	\vspace{0.1in}
	{
	\begin{algorithmic}[1]
		\STATE Sort content copies in $S$ obtained in Knapsack Storage: Part 1 (Figure \ref{fig:knapsack_storage_Part 1}) in increasing order of content index.   
		\STATE Divide each file into $a$ sub-files of equal size. 
		\STATE Store sub-file copy ranked $l$ in the ordered sequence on cache $((l-1)\mod m + 1)$.
	\end{algorithmic}}
	\vspace{0.1in}
	\hrule
	\caption{Knapsack Storage: Part 2 -- \sl Determines which contents to store on each cache.}
	\label{fig:knapsack_storage_Part 2}
\end{figure}

The following example illustrates Knapsack Storage:Part 2.

{\textit{Example:} Consider a system consisting of four caches, each with $1$ unit memory and can serve at most $2$ users as long as output data not exceed $1$ unit. Say the solution for Figure \ref{fig:knapsack_storage_Part 1} gives $x_1 = x_2 = x_3 = 1$ and 0 otherwise, and  $w_1 = 2$, $w_i= 1, \forall i\in \{2,3,...,n\}$. Figure \ref{fig:knapsack_storage_Part 2_example} illustrates Knapsack Storage: Part 2.
\begin{figure}[h]
\begin{center}
	\begin{center}
		Sorted S: 
		\begin{tabular}{| c | c | c | c |}
			\hline
			$1$ & $1$ &  $2$  & $3$   \\
			\hline
		\end{tabular}\\
	(i)

		\text{}

		Divide each file into 2 sub-files 
		{\footnotesize
			\begin{tabular}{|c|c|c|c|c|c|c|c|}
				\hline
				$1a$&$1b$&$1a$&$1b$&$2a$&$2b$&$3a$&$3b$\\
				\hline
		\end{tabular}}\\
	(ii)

		\text{}

		\begin{tabular}{| c ||c| c |}
			\hline
			Cache $1$ & $1a$ & $2a$ \\ 
			\hline
			Cache $2$& $1b$ & $2b$ \\ 
			\hline 
			Cache $3$&$1a$ & $3a$ \\  
			\hline
			Cache $4$& $1b$ &$3b$ \\  
			\hline
		\end{tabular}\\\vspace{0.05in}
	
	(iii)
	\end{center}
	\caption{Illustration of Knapsack Storage: Part 2 for a system with four caches each with $a=2$ and $k=1$.  Here, we assume that the  Knapsack Storage: Part 1 is $x_1 = x_2 = x_3  = 1$ and 0 otherwise, and  $w_1 = 2$, $w_i= 1, \forall i\in \{2,3,...,n\}$.  Hence, S contains 1 two ($w_1$) times and 2 and 3 one ($w_2=w_3$) time(s). (i) sorted S (ii) the files in sorted S are divided into two equal ($a$ and $b$) parts without changing the order in sorted S, and (iii) sub-file copy ranked $l$ in the ordered sequence is stored in cache $((l-1)\mod 4 + 1)$. For example sub-file $2a$ is ranked $5^{th}$ in (ii) and is stored in Cache $((5-1)\mod 4 + 1)=1$. }\label{fig:knapsack_storage_Part 2_example}
\end{center}
\end{figure}}


\textbf{Matching Policy: Match Least Popular} --
The next task is to match requests to caches. The key idea of the Match Least Popular policy is to match requests for the less popular contents before matching requests for the more popular contents. Please refer to Figure \ref{fig:match_least_popular} for a formal description of the Match Least Popular policy.

{Since each content is divided into $a$ sub-files, we divide each request for a content into $a$ sub-requests and allocate these sub-requests to caches storing the corresponding sub-files. Since each cache can serve $a$ sub-requests, we make $a$ copies of each cache and find a matching between the set of sub-requests and the set of cache copies.
	
	We index sub-files and cache copies as follows:
\begin{itemize}
	\item[--] The $j^{\text{th}}$ sub-file of File $i$ is indexed {$(i-1)a+j$,} for $1 \leq i \leq n, \ \text{and }1 \leq j \leq a.$ 
	
	\item[--] The $j^{\text{th}}$ copy of cache $i$ is indexed {$(i-1)a+j$,} for $1 \leq i \leq m, \ \text{and }1 \leq j \leq a.$ 
	
\end{itemize}
Figure \ref{fig:match_least_popular} describes the Match Least Popular policy.

}
\begin{figure}[H]
	\hrule
	\vspace{0.1in}
	{
	\begin{algorithmic}[1]
		\STATE initialize $i = an$, set of idle caches copies $= \{1,2,..., am\}$. 
		\IF {the number of requests for Sub-file $i$  is more than the \\
			\indent \hspace{0.09in}  number of idle cache copies storing Sub-file $i$,}
		\STATE goto Step 8.
		\ELSE 
		\STATE match requests for Sub-file $i$ to idle cache copies storing Sub-file $i$, chosen uniformly at random. 
		\STATE update the set of idle cache copies. 
		\ENDIF
		\STATE $i = i-1$, goto Step 2.
	\end{algorithmic}}
	\vspace{0.1in}
	\hrule
	\caption{Match Least Popular -- \sl Matches requests to caches.  Here, idle cache copies means cache copies not allocated to any request.}
	\label{fig:match_least_popular}
\end{figure}

{\textbf{Service Policy} --
All the sub-requests matched to copies of Cache $i$ are served by the Cache $i$ and all unserved requests are served by the server.}

\begin{theorem}\label{thm:zipf12u} {\em [Upper bound]}
	Consider a system with $n$ files with popularity following the Zipf distribution with parameter $\beta \in (1,2)$,  and $m=n/c$ ($c>0$, $c$ is a constant) caches of size $k $ units  each. Every cache is capable of  serving at most $a$ requests  and the total data served by a cache cannot exceed $1$ unit. The system receives a batch of  $r=\rho m $ ($0<\rho<1$, $\rho$ is a constant) i.i.d. request.
		Let $R_{z_2}$ be the central server's transmission rate for our policy described above. Then, 
		\begin{eqnarray*}
			\mathbb{E}[R_{z_2}]
			= \begin{cases}
				\OO\big(n^{2-\beta}\big)&\text{if } k <  c - \Theta(1), \text{ } \forall a, \\
				\OO\big(n^{(2-\beta-\gamma) / \beta }\big) &\text{if } k = c, \ a=m^{\gamma}, \gamma \in [0,1],\\
				\OO(1)&\text{if } k>c+\Theta(1), \text{ } \forall a.
			\end{cases}
		\end{eqnarray*}
\end{theorem}

 From Theorems \ref{thm:zipf12} and \ref{thm:zipf12u}, we conclude that if content popularity follows the Zipf's distribution with $1 < \beta < 2$, the Knapsack Storage + Match Least Popular policy is order-optimal in the class of policies which do not use coded placement or delivery.

 {Theorems \ref{thm:zipf12} and \ref{thm:zipf12u}} show that for $k\leq c - \Theta(1)$, the central server's transmission rate for any value of $a$ and any policy is $\Theta(n^{2-\beta})$. On the other hand, Theorem \ref{thm:zipf12u} shows that for $k \geq c+\Theta(1)$, there exists a storage/service policy for which the transmission rate of the central server for $a = 1$ (no resource pooling) is $\OO(1)$ with high probability. We thus conclude that for the Zipf popularity distribution with $\beta \in (1,2)$, there is no order-wise benefit of resource pooling ($a > 1$) for $k \le c - \Theta(1)$ and $k \geq c+\Theta(1)$. In addition, for $ k = c$, we need $a$ to be at least poly$(n)$, more specifically $\Omega(n^{2-\beta})$, to bring the server transmission rate to a constant. Therefore, the benefits of resource pooling in the case where content popularity follows the Zipf distribution with $\beta \in (1,2)$ are limited. Note that this is in sharp contrast to the results for the case when $\beta \in [0, 1)$, where a small amount of resource pooling, in particular $a = \Theta(\ln n)$, is sufficient to bring down the central server's transmission rate to $\OO(1)$.
 
\remove{
{\color{blue}
\begin{remark}
	It can be shown that the information-theoretic lower bound for our setting with $\beta>1$ is 
	\begin{eqnarray*}
		\mathbb{E}[R_{z_2}^*]
		= \begin{cases}
			\Omega\big(n^{2-\beta}\big)&\text{if } k <  c - \Theta(1), \text{ } \forall a, \\
			\Omega(0)&\text{if } k>c+\Theta(1), \text{ } a=1.
		\end{cases}
	\end{eqnarray*}
	From Theorem \ref{thm:zipf12u}, we conclude that our KS+MLP policy is order wise optimal except a small region around $k=c$. This is quite contrasting with \cite{maddah2014fundamental}, which shows coded caching achieves order optimal results. 
	
\end{remark}
}
 }
 
 \begin{remark}
	{Theorems \ref{thm:zipf12} and \ref{thm:zipf12u} also hold for $\beta \geq 2$ and $\rho=1$.}
\end{remark}


	\section{Simulation Results}
\label{sec:simulations}

In Section \ref{sec:results}, we evaluated the performance of our policies asymptotically, $i.e.,$ as $m\rightarrow \infty$. In this section, we simulate the system for finite values and compare the performance of various placement and delivery policies. {We begin by simulating the performance of the Proportional Placement + Optimal Matching Routing (PP+OMR) policy described in Section~\ref{zipf} and whose asymptotic performance for the case of $\beta < 1$ was presented. Recall that, in the PP+OMR policy, the service policy (OMR) is based on the maximal matching between the set of servers and the set of sub-requests in each time-slot. Since this is an expensive operation with time complexity {$\OO(k\min\{(an)^{2.376},(an)^{2+\beta} \})$}, we propose three other computationally inexpensive service policies and evaluate their performance via simulations. Our motivation is to determine if the benefits of resource pooling extend to the computationally inexpensive service policies as well. 

The first alternate service policy is the Match Least Popular (MLP) policy, which is described in  Section \ref{zipfb}.} The second service policy as Online Randomized Routing (ORR). Let requests be indexed from 1 to $r$. Starting from the first request, this policy sequentially allocates requests to caches as follows: each request is divided into $a$ sub-requests, one each for the $a$ sub-files. Each sub-request is then allocated to any cache which stores the requested sub-file and can accommodate one more request, chosen uniformly at random. The third service policy, called Online Least-loaded Routing (OLLR) also allocates requests in a sequential manner. The difference between the ORR and the OLLR policy is that the OLLR policy allocates each sub-request to the least loaded cache which stores the requested sub-file and can accommodate one more request, breaking ties uniformly at random.  The time complexity of these new policies is $\OO(akn^{1+\beta})$. Note that the OMR and MLP service policies are offline policies, which do the cache assignment based on the entire collection of requests. The ORR policy and OLLR policies are online policies, which serves requests in an arbitrary sequential order. We combine the Proportional Placement (PP) policy with each of these delivery policies and compare their performance, denoting the corresponding policies as PP+OMR policy, PP+MLP policy, PP+ORR policy and PP+OLLR policy.

We simulate the distributed content delivery network described in Section \ref{sec:setting} to compare the performance of the PP+MLP, PP+ORR, PP+OLLR and PP+OMR policies as a function of various system parameters like the storage capacity of the caches ($k$), the maximum number of users each cache can serve in a time-slot ($a$), {and the product ($ak$).} 
 We focus on the case where the number of caches ($m$) is equal to the number of files ($n$) and the content popularity follows the Zipf distribution with parameter $\beta = 0.3$. For each set of system parameters, we report the mean transmission rate averaged over 1000 iterations.


\begin{figure}[t]
	\centering

	\begin{minipage}{0.46\textwidth}
		\centering
	\includegraphics[width = 1.0\textwidth]{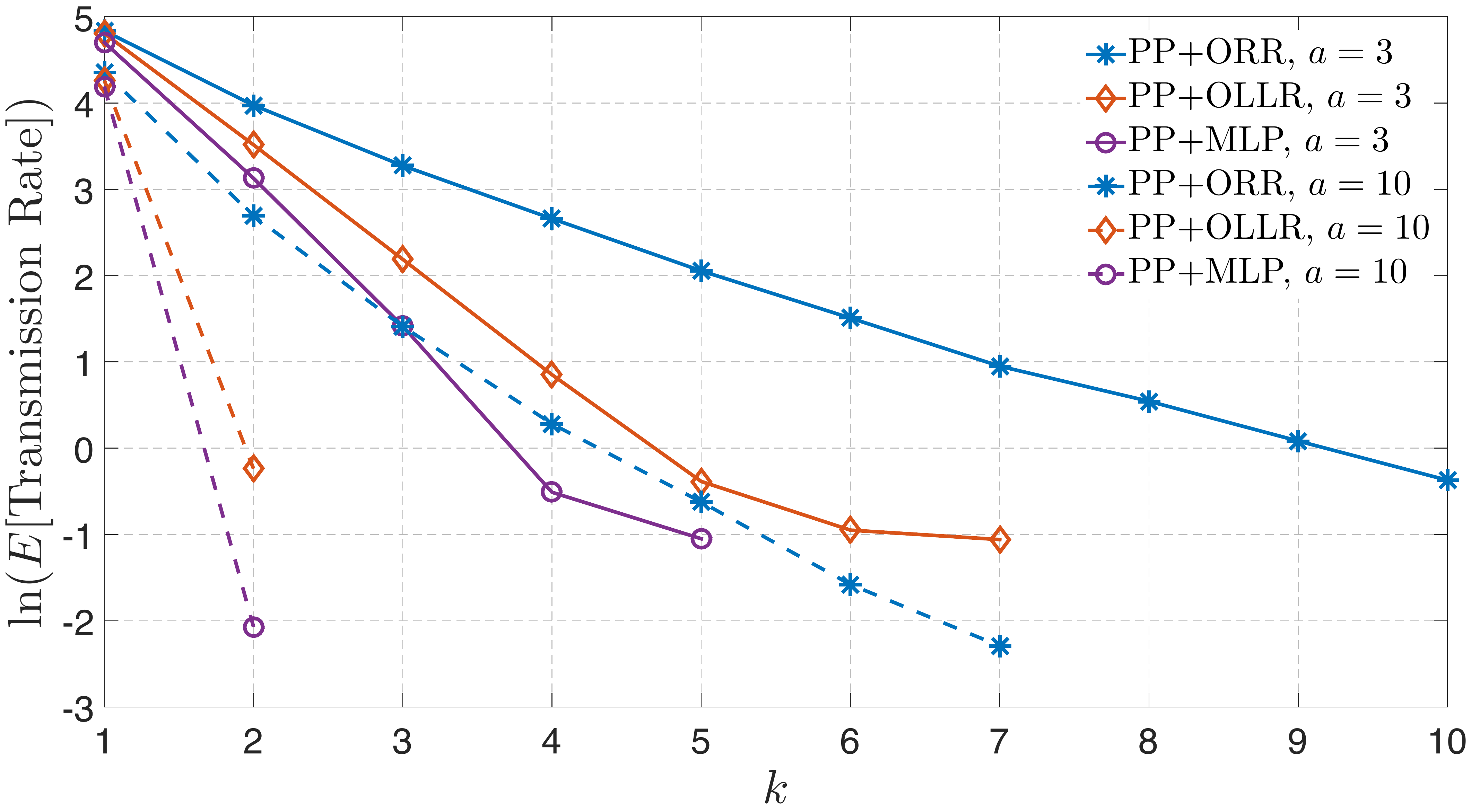}\\(i)
	\end{minipage}
   \begin{minipage}{0.46\textwidth}
   	\centering
   	\includegraphics[width = 1.0\textwidth]{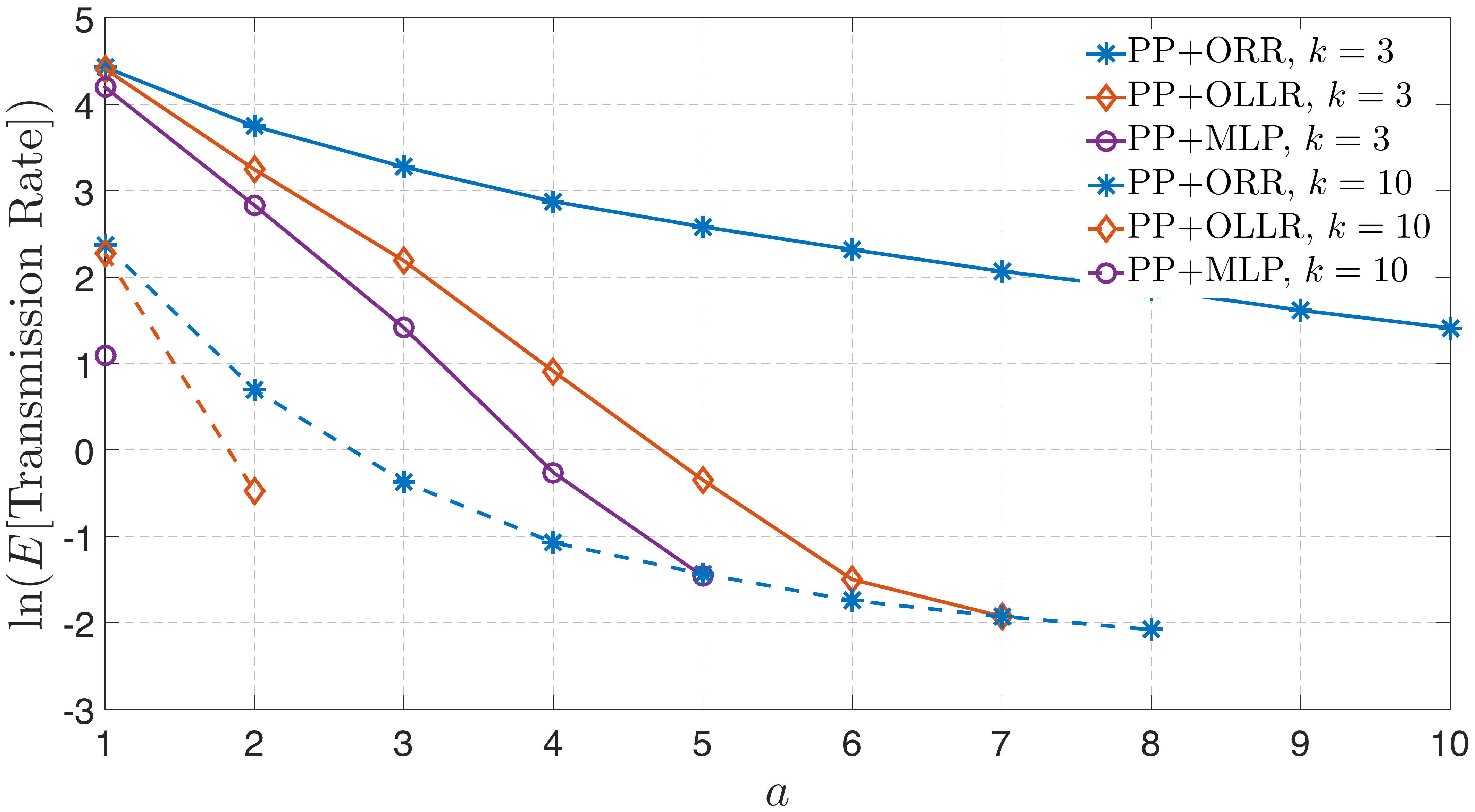}\\(ii)
   \end{minipage}
\begin{minipage}{0.46\textwidth}
	\centering
	\includegraphics[width = 1.0\textwidth]{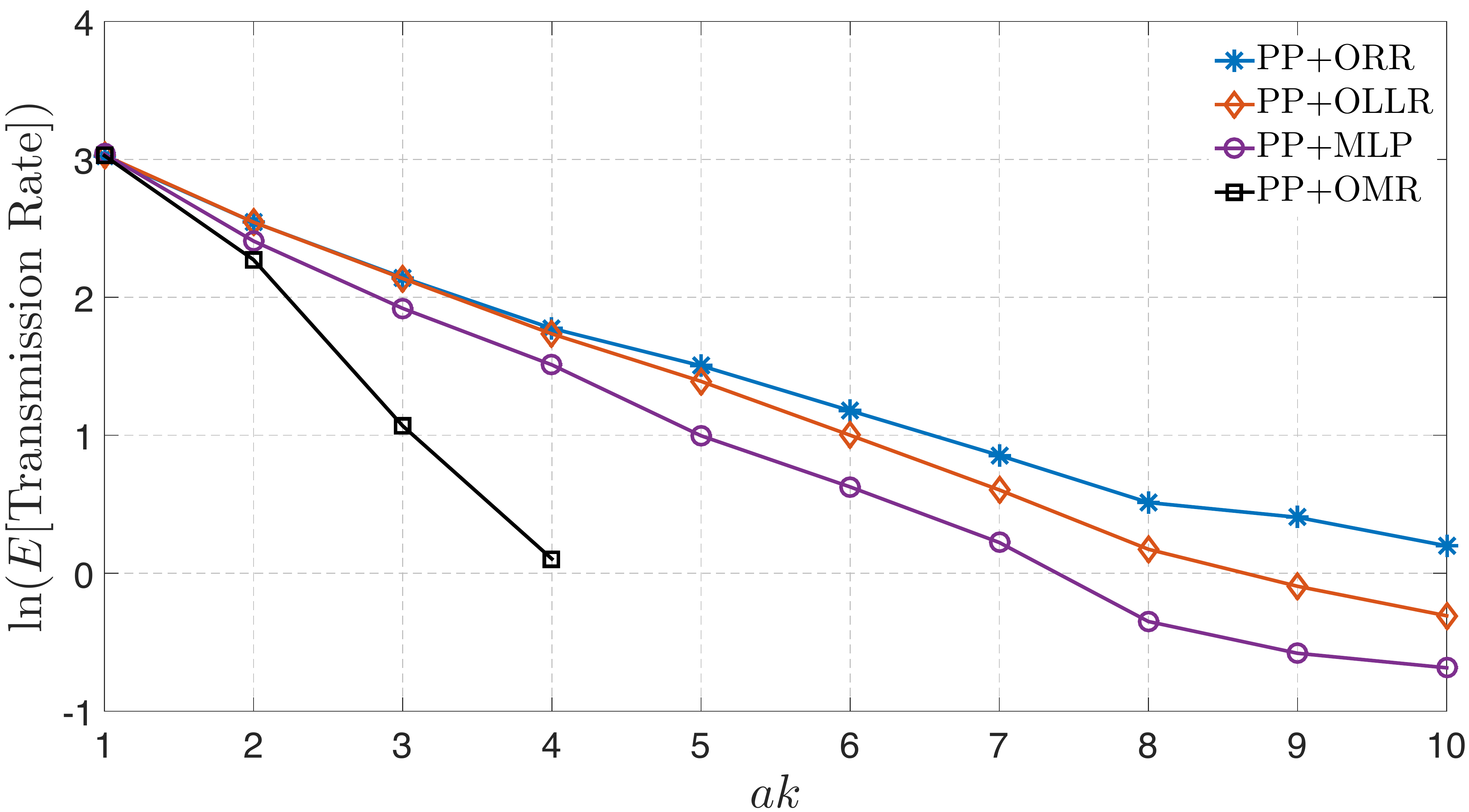}\\(iii)
\end{minipage}
	\caption{Plot of the mean transmission rate for the PP+ORR, PP+OLLR, PP+MLP, and PP+OMR policies as a function of (i) the storage capacity of the cache ($k$), (ii) the maximum number of users each cache can serve in a time-slot ($a$), and (iii) the product $ak$. The system parameters for Figures (i) and (ii) are $n=1000$ files  and $m=1000$ caches,  and  $r=800$ requests. The system parameters for Figure (iii) are $n=100$ files and $m=100$ caches,  and  $r=80$ requests.  Due to high time complexity, we plot the performance of the PP+OMR only for $n=m=100$ (Figure (iii)). In all these plots, the transmission rate is 0 after a few values. Since, we plot the log of expected transmission rate, the corresponding lines are terminated earlier.}\label{fig:all}
		\hspace{-0.4in}	
\end{figure}


Theorem~\ref{thm:zipf01} states that the upper bound on the transmission rate for the OMR service policy decreases exponentially with the product of the storage capacity of the cache ($k$) and the maximum number of users each cache can serve in a time-slot $(a)$.
In Figure \ref{fig:all}(i), we plot the mean transmission rate for PP+MLP, PP+ORR, and PP+OLLR policies as a function of the storage capacity of each cache ($k$), for a system where the number of files and caches is 1000, and a batch of 800 requests is  served.  We see that for a fixed value of $a$, the transmission rate decreases exponentially with $k$. In addition, for a fixed value of $k$, the performance of all policies improves with increase in $a$.


\remove{\begin{figure}[h]
	\centering
	
	\begin{minipage}{0.45\textwidth}
		\centering
	\includegraphics[width = 1.0\textwidth]{na}
	\caption{Plot of the mean transmission rate for PP+MLP, PP+ORR, and PP+OLLR policies as a function of the maximum number of users each cache can serve in a time-slot ($a$), for a system with $1000$ files ($n=1000$) and $1000$ caches ($m=1000$),  and a batch of 800 requests ($r=800$) is served { and a batch of 800 requests ($r=800$), each request follows the Zipf distribution with $\beta=0.3$ independently.}\vspace{0.2in}}\label{fig:orrollr_a}
	\end{minipage}
\hspace{0.1in}
	\begin{minipage}{0.45\textwidth}
		\begin{center}
		\includegraphics[width=1.0\textwidth]{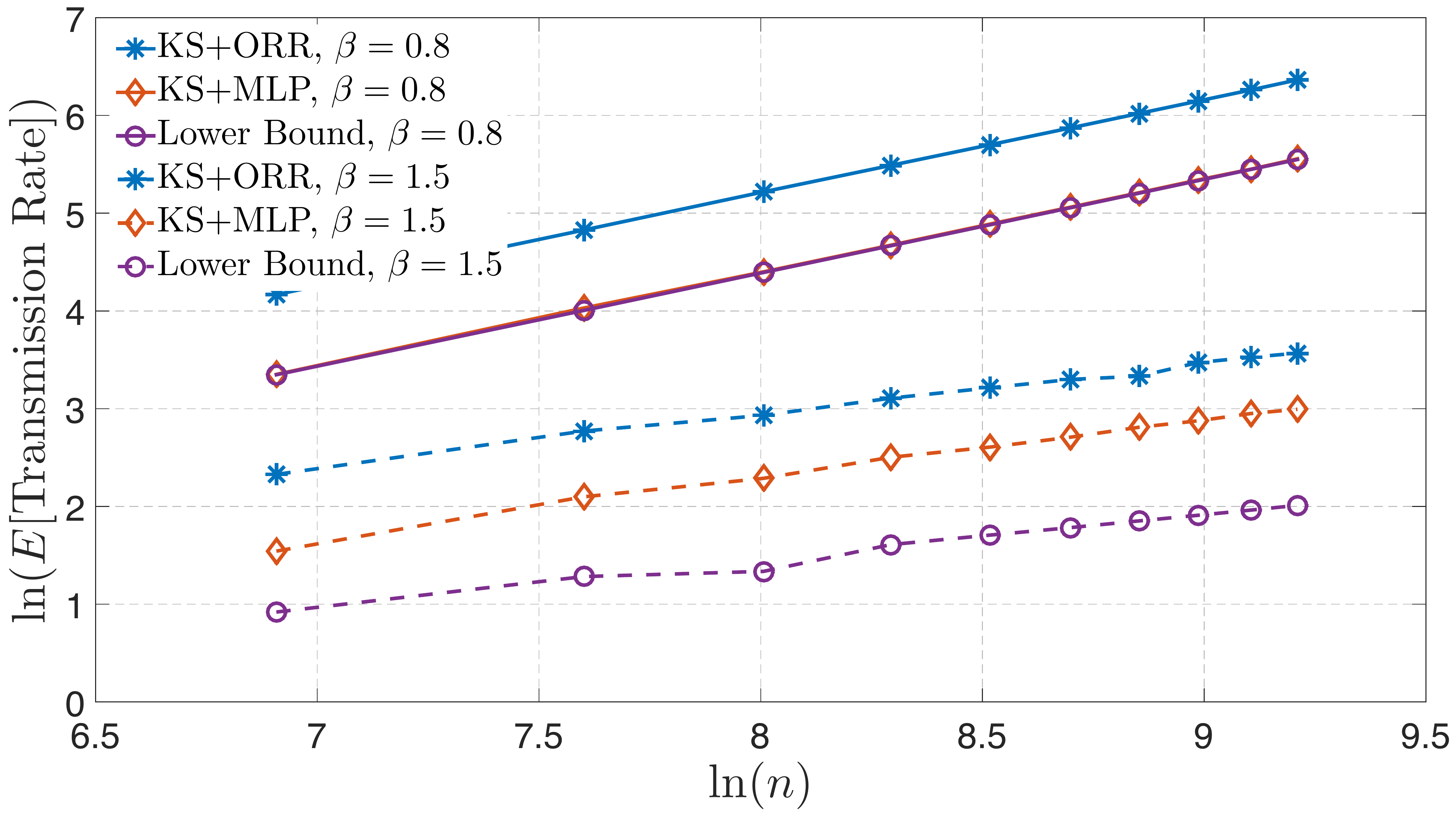}
		\caption{\color{red}\sl Plot of the mean transmission rate for the KS+MLP policy, KS+ORR policy and the lower bound on the expected transmission rate as a function of the number of files ($n$), for a system where the number of caches ($m$) is one fifth of the number of files ($n = 5m$), and each cache can store three files (${k}=3$). \label{fig:vary_size_1}}
	\end{center}
	\end{minipage}
	
\end{figure}}

\begin{figure}[t]
	\centering
	\begin{minipage}{0.46\textwidth}
		\centering
		\includegraphics[width = 1.0\textwidth]{nn}\\(i)
	\end{minipage}
	\begin{minipage}{0.46\textwidth}
		\centering
		\includegraphics[width = 1.0\textwidth]{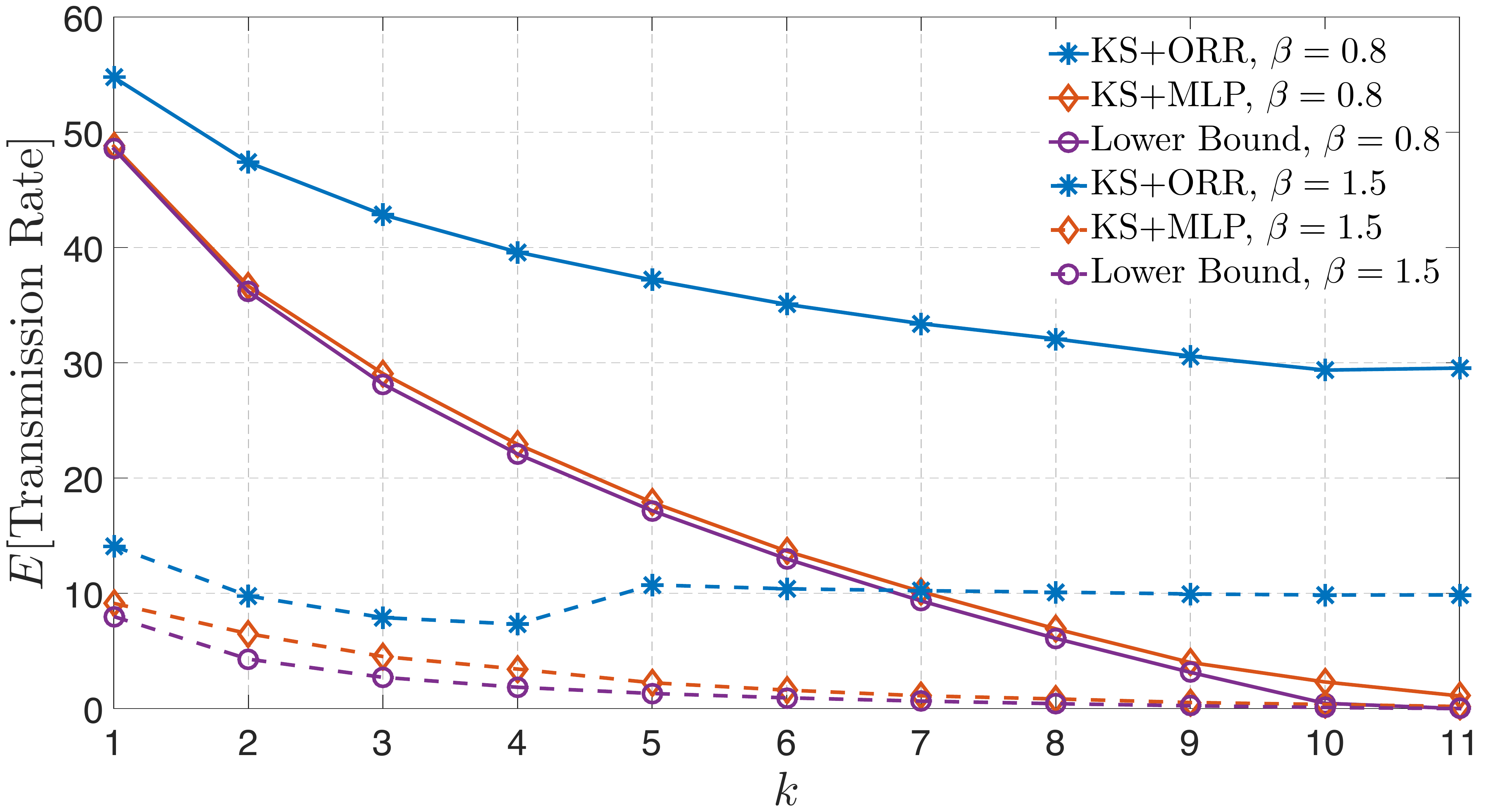}\\(ii)
	\end{minipage}
	\begin{minipage}{0.46\textwidth}
		\centering
		\includegraphics[width = 1.0\textwidth]{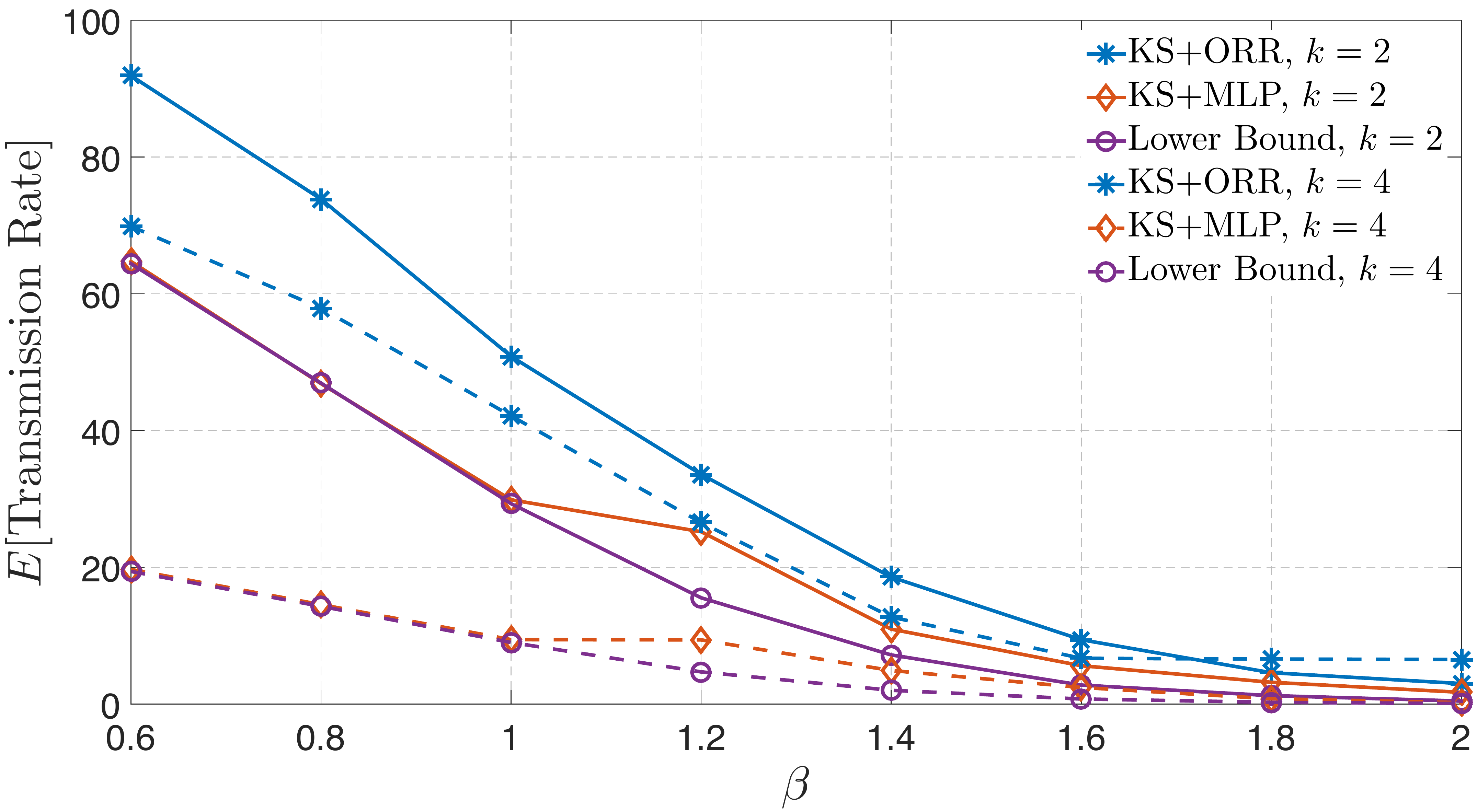}\\(iii)
	\end{minipage}
	\caption{Plot of the mean transmission rate for the KS+MLP policy, KS+ORR policy and the lower bound on the expected transmission rate (i) as a function of the number of files ($n$), for a system where the number of caches ($m$) is one fifth of the number of files ($n = 5m$), and each cache can store three files (${k}=3$), (ii) as a function of storage capacity per cache (${k}$) for a system with $n=1000$ files and $m=100$ caches,	and (iii) as a function of Zipf parameter ($\beta$) for a system with $n=1000$ files and $m=200$ caches. In all the figures, the number of requests ($r$) in a time-slot is equal to the number of caches ($m$).}	\label{fig:knapsack}
	\hspace{-0.4in}	
\end{figure}
In Figure \ref{fig:all}(ii), we plot the mean transmission rate for PP+MLP, PP+ORR, and PP+OLLR policies as a function of the maximum number of users each cache can serve in a time-slot ($a$), for a system where the number of files and caches equal to 1000, and a batch of 800 requests is served. From the plot, we see that for a fixed value of $k$, the transmission rate decreases exponentially with $a$. In addition, for a fixed value of $a$, the performance of all policies improves with increase in $k$. In Figure \ref{fig:all}(iii), we plot the mean transmission rate for PP+MLP, PP+ORR, PP+OLLR and PP+OMR policies as a function of the product $ak$, for a system where the number of files and the number of caches is 100, and a batch of 80 requests is served. Note that out of all the combinations of $a$ and $k$ which lead to the same product $(ak)$, we consider the values of $a$ and $k$, which give the minimum mean transmission rate.  From the plot, we see that the transmission rates for all the four policies decrease exponentially with $ak$, with the offline OMR and MLP policies performing better than the online OLLR and ORR policies as expected. We thus conclude that benefits of resource pooling also extend to the computationally inexpensive service policies PP+MLP, PP+ORR, and PP+OLLR.

Next, we evaluate the performance of Knapsack Storage + Match Least Popular (KS+MLP) policy described in Section~\ref{zipfb} whose asymptotic performance for the case of $\beta \in (1, 2)$ was presented in Theorem~\ref{thm:zipf12u}. We also simulate the performance of Knapsack Storage + Online Random Routing (KS+ORR), where the ORR delivery policy is as described earlier. We compare the performance of the KS+MLP and KS+ORR policies as well as the lower bound on the performance of all uncoded policies derived in Theorem~\ref{thm:zipf12}.

As before, we simulate a distributed cache system with content popularity following the Zipf distribution to understand how the performance of the KS+MLP and KS+ORR policies depends on various parameters like number of contents $(n)$, number of caches ($m$), storage capacity per cache $({k})$, and Zipf parameter $(\beta)$. We focus on the case where the number of requests per time-slot is equal to the number of caches. For each set of system parameters, we report the mean transmission rate averaged over 10000 iterations.

In Figure \ref{fig:knapsack}(i), we plot the mean transmission rates for  the KS+MLP and KS+ORR policies as well as the lower bound on the expected transmission rate as a function of the number of contents ($n$), for a system where the number of caches ($m$) is one fifth of the number of contents ($n = 5m$), and each cache can store three contents (${k}=3$). In this regime, Theorems~\ref{thm:zipf12} and \ref{thm:zipf12u} suggest that the mean transmission rate for the KS+MLP policy is $\OO(n^{2-\beta})$ and the lower bound on the expected transmission rate is $\Omega(n^{2-\beta})$. We see that the mean transmission rates for the KS+MLP and KS+ORR policies as well as the lower bound follow the expected trend.}
%
In Figure \ref{fig:knapsack}(ii), we plot the mean transmission rates for KS+MLP and KS+ORR policies and the lower bound on the expected transmission rate as a function of the storage per cache ({$k$}) for a system with 1000 contents ($n=1000$) and 100 caches ($m=100$). As expected the mean transmission rates for KS+MLP and KS+ORR policies and the lower bound on the expected transmission rate are decreasing functions of $k$. We see that the mean transmission rate for the KS+MLP policy is very close to the lower bound on the expected transmission rate while the KS+ORR\footnote{The plot of the expected rate of the KS+ORR policy in Figure \ref{fig:knapsack}(ii) has a small unexpected jump. As the memory increases the number of files stored in the cache is increasing but the service policy is online and it matches requests randomly. So the chance of high popular file getting requested first is high and can lead to a blocking of the caches hosting the lower popularity files. This leads to the misbehavior in the performance. We also observe similar behavior in Figure \ref{fig:knapsack}(iii) for large values of $\beta$. } policy performs significantly worse.  In Figure \ref{fig:knapsack}(iii), we plot the mean transmission rates for KS+MLP and KS+ORR policies and the lower bound on the expected transmission rate as a function of the Zipf parameter $\beta$. We simulate a system with 1000 contents ($n=1000$) and 200 caches ($m=200$) for two different values of storage per cache. As expected, the mean transmission rates for KS+MLP and KS+ORR policies and the lower bound on the expected transmission rate are decreasing functions of $\beta$. We see that the mean transmission rate for the KS+MLP policy is very close to the lower bound on the expected transmission rate. 
Note that our simulations are for the case where the number of caches ($m$) is equal to the number of files ($n$). But, similar tendency is maintained, when the parameter setting is changed in cases with $n > m$ or $n < m$.

	\section{Conclusions and future work}
\label{sec:conclusions}
In this work we focus on a content delivery system consisting of a central server which communicates over an error-free broadcast channel with multiple co-located caches, each with limited storage and service capabilities. A key feature of this work is that we allow resource pooling across caches which allows a request to be simultaneously served by multiple caches. We propose two policies, ($\emph{i}$) PP+OMR (for $0\leq \beta <1$, where $\beta$ is the Zipf parameter) and  ($\emph{ii}$) KS+MLP (for $1< \beta <2$), and characterize their asymptotic performance. We also derive fundamental lower bounds on the optimal server transmission rate for our system. We conduct extensive simulations to compare the performance of our proposed schemes as well as other natural online and computationally-efficient variants. 
	More specifically, if the popularity does not vary drastically
		across contents, then resource pooling leads to an order wise reduction in central server transmission rate as the system size grows. On the other hand, if the content popularity is skewed, the central server transmission rate is
		of the same order with and without resource pooling.
In this paper, we restrict our analysis to uncoded placement policies. In the future, we plan to study the advantages of coded placement, both in terms of improving system performance as well as providing reliability against cache failures.
	\section{proofs} \label{sec:proof}

\subsection{Proof of Theorem \ref{thm:zipf01}}

 {We characterize the performance of our storage/service policy discussed in Section \ref{sec:results}. 
 	
 \textbf{Case 1:} $c-k=\Theta(1)$: Theorem \ref{thm:zipf01}(a) is trivial in this case.
 
 We use the following lemmas to prove Theorem \ref{thm:zipf01}(b).
 	\begin{lemma} \label{lemma:atleastonce}
 	For $X=\text{Bin}(n,p=\frac{c}{n})$, such that $c>0$ is a constant, then, as $n\rightarrow \infty$,
 	$
 	\mathbb{P}(x\geq 1 )=1-e^{-c}.$ 

 \end{lemma}

\begin{lemma}\label{lemma:leastpopular}
	Let $P=\{p_1,p_2,...,p_n\}$ be the Zipf distribution with parameter $\beta$. 
	Then for $\beta\in[0,1)$, $p_i\geq\frac{1-\beta}{n}$ $\forall i$. 
\end{lemma}
%

\begin{proof} (Proof of Theorem \ref{thm:zipf01}(b)) Let $c-k=x$. Any uncoded storage policy can't store more than $mk$ units of data at least once due to memory constraint. From Lemmas \ref{lemma:atleastonce} and \ref{lemma:leastpopular}, the unstored $n-mk=xm$ units of data is requested at least once with probability $\geq 1-e^{-\frac{(1-\beta)\rho}{c}}$. Hence, 
	$\EE[R_{z_1}^*]\geq x\Big(1-e^{-\frac{(1-\beta)\rho}{c}}\Big)m.$
\end{proof}
\hspace{0.25in}

\textbf{Case 2:} {$k-c=\Theta(1)$}: This proof uses ideas from the proof of Proposition 1 in \cite{LLM12} which looks at the setting where each cache can serve at most one request. We first compute a lower bound on the probability that there exists a fractional matching between the set of sub-requests and the caches such that the total data served by each cache is less than $1$ unit. By the total unimodularity of adjacency matrix, the existence of a fractional matching implies the existence of an integral matching \cite{LLM12}. Since each sub-request is for $1/a$ units of data, to ensure that the total data served by each cache is less than $1$ unit in the integral matching, each cache will be allocated not more than $a$ sub-requests. The integral matching thus satisfies the restrictions discussed in Section~\ref{sec:setting} and therefore is a valid allocation of requests to caches. 

\begin{proof} (Proof of Theorem \ref{thm:zipf01}(a)): Recall our storage policy in Section \ref{sec:results}. We divide each file into $a$ sub-files of equal size and the number of caches storing a sub-file is proportional to its popularity, $i.e.$, for all $i$, each sub-file of Content $i$ is stored on $d_i\approx mkp_i$ caches. 
	  \remove{Let $p^*=\frac{{1-\beta}}{n}$. From Lemma \ref{lemma:leastpopular}, $p_i\geq p^*\ \forall i.$}
	   Recall that the number of requests for Content $i$ in a time-slot, denoted by $b_i$, is Bin($r,p_i$), where, $p_i={p_1}{i^{-\beta}}$. We treat each request for a file as $a$ sub-requests, one for each of the $a$ sub-files. 
	
	For each sub-file of Content $i$ and each of the corresponding $b_i$ sub-requests, we split each sub-request into $d_i$ sub-sub-requests of size $\frac{1}{d_i}$ each. 
	Let $\partial s$ denote the set of sub-files stored on Cache $s$. For each $i \in \partial s$, we associate $b_i$ sub-sub-requests for Sub-file $i$ to Cache $s$. This allocation leads to a fractional matching where the the total data served by each cache is less than $1$ unit if $\forall s\in \{ 1,2,...,m\}$,	
	$
	\sum_{i\in \partial s}\frac{b_i}{d_i}\leq a \implies
	\sum_{i\in \partial s}\frac{b_i}{ad_i}\leq 1,$ and,	
	\begin{align}
	&\mathbb{P}\bigg(\sum_{i\in \partial s}\frac{b_i}{ad_i}>1\bigg)\leq \inf_{s>0}\text{ }\frac{\mathbb{E}\bigg[e^{s\sum_{i\in \partial s}\frac{b_i}{ad_i}}\bigg]}{e^s}\nonumber\\  
	&\leq \inf_{s>0}\text{ } e^{-s} \prod_{i\in \partial s}
	e^{r\ln\big(p^*e^\frac{s}{amkp^*}+1-p^*\big)} 
	=e^{-\frac{ak(1-\beta)\rho}{c}h(\frac{1}{\rho})},\nonumber
	\end{align}
	where, $p^*=\frac{{1-\beta}}{n}$, $h(x)=x\ln{x}-x+1$ is the Cramer transform of a unit Poisson random variable. \remove{Here, the second inequality is due to $b_i$'s are negatively associated random variables \cite{dubhashi1996balls}.}
	 {If we broadcast the data of the caches with $\sum_{i\in \partial s}\frac{b_i}{ad_i}>1$, the requests for all other files can be served via the caches since these requests can be matched to caches while ensuring $\sum_{i\in \partial s}\frac{b_i}{ad_i}\leq1$ for each of them. Hence, there exists a fractional matching which, as discussed before, implies the existence of an integral matching for the remaining caches and requests. Therefore, the expected transmission  rate is $\mathbb{E}[R_{z_1}]\leq mk \mathbb{P}\bigg(\sum_{i\in \partial s}\frac{b_i}{ad_i}>1\bigg)\leq mke^{-\frac{ak(1-\beta)\rho}{c}h\big(\frac{1}{\rho}\big)}. $}
	   \remove{Taking the Union bound over all $m$ caches,
	$
	\mathbb{P}\big(\text{matching exists}\big) \geq 1-me^{-\frac{ak(1-\beta)\rho}{c}h\big(\frac{1}{\rho}\big)}.
	$
	Then, 
	\begin{align*}
	\mathbb{E}[R_{z_1}] &\leq 0 \times \mathbb{P}\big(\text{matching exists}\big) + r \times \big(1- \mathbb{P}\big(\text{matching exists}\big)\big) \leq rme^{-\frac{ak(1-\beta)\rho}{c}h\big(\frac{1}{\rho}\big)}.
	\end{align*}}	
\end{proof}

 Next, we obtain a lower bound on the transmission rate for any storage/service policy for $0\leq\beta<1$.
 We use the following lemma to prove Theorem~\ref{thm:zipf01}(b).
 \begin{lemma}
 	\label{lemma:total_memory}
 	In a system with $n$ files of size $1$ unit and $m$ caches of size $k$ units each, at most $\frac{m}{2}$ {units of data}  can be stored at least $2k$ times each. (Note: Proof by contradiction)
 \end{lemma}
 
 To prove Theorem\ref{thm:zipf01}(b), instead of lower bounding the expected transmission rate of the original system, we lower bound the expected transmission rate for an alternative system, which is less restrictive, and therefore, more powerful than the original system described in Theorem \ref{thm:zipf01}. 
 
 \begin{proof} (Proof of Theorem \ref{thm:zipf01}(b)): In the original system the total output rate of each cache is restricted to at most $1$ unit, and each cache can serve at most $a$ requests. We consider an alternative system (System A), in which we allow each cache to serve multiple requests in each time slot, as long as it serves at most $a$ requests for each content stored in the cache.  Let $\mathbb{E}[R^*_A]$ be the expected transmission rate in the alternative system. Since, the alternative system is less restrictive than the original system, it follows that $\mathbb{E}[R_{z_1}^*]\geq \mathbb{E}[R^*_A]$.
 	
 	
 	From Lemma \ref{lemma:total_memory}, we can conclude that there exist at least $n-\frac{m}{2}$ {units of data} which stored not more than $2k$ times each. Each of these $n-\frac{m}{2}$ {units of data} is requested with probability $\geq p_n$. Consider another system (System B) in which these $n-\frac{m}{2}$ {units of data} are requested with uniform probability $p_n=\frac{p_1}{n^\beta}$. Let $\mathbb{E}[R^*_B]$ be the expected transmission rate in System B. Then, $\mathbb{E}[R_A^*]\geq \mathbb{E}[R^*_B]$. From \remove{\eqref{eq:3.1} and \eqref{eq:3.2}}Lemma \ref{lemma:leastpopular}, we know that,
 	$
 	p_n \geq \frac{{1-\beta}}{n}. \nonumber
 	$
 	Now, consider a new system (System C), in which these $n-\frac{m}{2}$ {units of data} are requested with probability $\frac{{1-\beta}}{n}$. Let $\mathbb{E}[R^*_C]$ be the expected transmission rate in System C. It follows that
 	$\mathbb{E}[R^*_B]\geq \mathbb{E}[R^*_C].$
 	
 	 The number of requests for Content $i$ in new system, denoted by $b_i'$, is Bin($r,\frac{1-\beta}{n}$).
   Then, for  large enough $n$, we have that,
 	\begin{align}
 	\mathbb{P}(b_j'>2ak)&=\sum_{i=2ak+1}^{r} {{r} \choose i}\bigg({\frac{1-\beta}{n}}\bigg)^i\bigg(1-\frac{1-\beta}{n}\bigg)^{{r}-i}\nonumber\\
 	&\geq e^{\frac{(1-\beta)\rho}{c}}\frac{\big(\frac{(1-\beta)\rho}{c}\big)^{2ak+1}}{2(2ak+1)!}. \nonumber
 	\end{align}
 	Since, there are at least \big($n-\frac{m}{2}$\big) such units, the expected central server's transmission rate can be lower bounded as follows: \\
 	 $\mathbb{E}[R_{z_1}^*]\geq  \mathbb{E}[R^*_C]
 	\geq \Big(n-\frac{m}{2}\Big)e^{\frac{(1-\beta)\rho}{c}}\frac{\Big(\frac{(1-\beta)\rho}{c}\Big)^{2ak+1}}{2(2ak+1)!}.$
 \end{proof}

\subsection{Proof of Proposition \ref{prop:zipf12}}
Next, we consider the case $\beta\in (1,2)$. We use the following lemma to prove Proposition \ref{prop:zipf12}.

\begin{lemma} \label{lemma:binomial}
For a Binomial random variable $X=\text{Bin}(m,p)$ s.t., $mp\geq 1$,
$
\lim_{m\rightarrow \infty} \mathbb{P}(x\geq  \lfloor mp\rfloor )\geq \frac{1}{2}.
$ \big(Recall: median($X$)$\geq\lfloor mp\rfloor$\big)
\end{lemma}

\begin{proof}{ (Proof of Proposition \ref{prop:zipf12})}
From Lemma \ref{lemma:binomial}, it is clear that, if {Bit $u$ of Content $i$ }s.t., ${r}p_i\geq a$ is stored on less than $\frac{{\lfloor r}p_i\rfloor}{a}$ caches, it will have to be fetched from the central server with probability $\geq 0.5$. Therefore, if we are interested in the order of the number of contents transmitted by the central server, storing {Bit $u$ of Content $i$ } on fewer than $\frac{{\lfloor r}p_i\rfloor}{a}$ caches is equivalent to not storing it at all. Therefore, to make the most use of the available cache memory, we restrict ourselves to the case where if caching policy decides to cache {Bit $u$ of Content $i$ } it is stored on at least $\max\big\{ \big\lfloor\frac{{r}p_i}{a}\big\rfloor,1\big\}$ caches.

If the caching policy decides not to cache {Bit $u$ of Content $i$ }, the central server will transmit this content if it is requested at least once in the batch of ${r}$ requests, $i.e.,$ with probability $\geq 1-(1-p_i)^{r}$.  Let $x_{i,u}=1$ imply that {Bit $u$ of Content $i$ } is cached and $x_{i,u}=0$ otherwise. 
\begin{align*}
&\mathbb{E}[{ R}^*_{NC}]=  \Omega\bigg(\min \sum_{i=1}^{n}\sum_{u=1}^{b_i}(1-x_{i,u})\Big(1-(1-p_i)^{r}\Big)\bigg)\\
&\text{s.t. } \sum_{i=1}^n\sum_{u=1}^{b_i}x_{i,u}\max\Big\{ \Big\lfloor\frac{{r}p_i}{a}\Big\rfloor,1\Big\}\leq mkb,\& \ x_{i,u} =\{ 0,1\} \text{, }\forall i. \\
&\text{Let }\text{O}_1^*=  \max \sum_{i=1}^{n}\sum_{u=1}^{b_i}x_{i,u}\Big(1-(1-p_i)^{r}\Big)\\
& \text{s.t. }  \sum_{i=1}^n\sum_{u=1}^{b_i}x_{i,u}\max\Big\{ \Big\lfloor\frac{{r}p_i}{a}\Big\rfloor,1\Big\}\leq mkb\text{ } \&\text{ }
  x_{i,u} =\{ 0,1\} \text{, }\forall i. 
\end{align*}
\begin{align*}
&\text{O}_1^* \leq \text{O}^*=  \max \sum_{i=1}^{n}\sum_{u=1}^{b_i}x_{i,u}\Big(1-(1-p_i)^{r}\Big) \\
& \text{s.t. } \sum_{i=1}^n\sum_{u=1}^{b_i}x_{i,u}\max\Big\{ \Big\lfloor\frac{{r}p_i}{a}\Big\rfloor,1\Big\}\leq mkb \text{ }  \& \text{ }0\leq x_{i,u} \leq 1 \text{, }\forall i.
\end{align*}


\begin{align*}
\therefore \mathbb{E}[R^*_{NC}]= \Omega\bigg( \sum_{i=1}^{n}b_i\Big(1-(1-p_i)^{r}\Big)-\text{O}_1^*\bigg)\\
 = \Omega\bigg( \sum_{i=1}^{n}b_i\Big(1-(1-p_i)^{r}\Big)-\text{O}^*\bigg).
\end{align*}
\end{proof}

\subsection{Proof of Theorem \ref{thm:zipf12}}
We use the following lemmas to prove Theorem \ref{thm:zipf12} and Theorem \ref{thm:zipf12u}.

\begin{lemma}
	\label{lemma:chernoff}
	For a Binomial random variable $X$ with mean $\mu$, by the Chernoff bound, $\forall$ $\delta\geq0$,
	\begin{align*}
		 \PP(X \geq (1+\delta) \mu) &\leq \Bigg(\frac{e^\delta}{(1+\delta)^{(1+\delta)}}\Bigg)^\mu, \\ 
		 \PP(X \leq (1-\delta) \mu)& \leq e^{-\delta^2 \mu/2}.
	\end{align*}
\end{lemma}
\begin{lemma}
	\label{lemma:Zipf_d}
	Let content popularity follow the Zipf distribution with Zipf parameter $\beta>1$. In a given time-slot, let $d_i$ be the number of requests for Content $i$. Let $E_1$ be the event that:
	\begin{enumerate}
		{\item[(a)] $ d_i \geq 1$ for $i=\OO(m^{\frac{1}{\beta}-\epsilon})$, where $\epsilon>0$ is arbitrarily small constant,}
		\item[(b)] $d_i \leq 2p_1 (\log m)^2$ for $n_1 < i \leq n_2$,
		\item[(c)] $d_i \leq \bigg(1+\dfrac{p_1}{4}\bigg) mp_i$ for $1 \leq i \leq n_1$,
	\end{enumerate}
	where $n_1$ and $n_2$ are as defined in Equation \ref{eq:weight1}.
	Then, 
	$\PP(E_1) = 1-\OO(n e^{-(\log m)^{2}}).$
\end{lemma}
\begin{proof} 	
	Since content popularity follows the Zipf distribution with Zipf parameter $\beta>1$,
	
	\begin{enumerate} 
		{\item[(a)] For all $i=\OO(m^{\frac{1}{\beta}-\epsilon})$, expected number of requests is $\OO(m^{\epsilon \beta})$, and
			$
				\PP\big(d_i < 1 \big) = \OO\Big(e^{-m^{\epsilon\beta}}\Big).
		$}
		
		\item[(b)] For all contents less popular than Content $n_1$, 
	$
			p_{i} \leq \frac{p_1 (\log m)^2}{m}.
	$
		Therefore, by the Chernoff bound (Lemma \ref{lemma:chernoff}), we have that, for $n_1 < i \leq n_2$,
		$
			\PP\big(d_i > 2p_1 (\log m)^2 \big) = \OO(e^{-(\log m)^{2}}).
		$
		\item[(c)] For $i \leq n_1$, $mp_i = \Omega((\log m)^{2})$, by the Lemma \ref{lemma:chernoff},
		$
			\PP\bigg(d_i > \bigg(1+\dfrac{p_1}{4}\bigg) mp_i \bigg) = \OO(e^{-mp_i}). 
	$
	\end{enumerate}
	\noindent Therefore, by the union bound over all contents, we have that, 
	$\PP(E_1) = 1-\OO(n e^{-(\log m)^{2}}).$
\end{proof}

\begin{proof} (Proof of Theorem \ref{thm:zipf12})\\
\textbf{Case 1:} $c-k>\Theta(1)$: Consider a new system with one cache of size $mk$ units which can serve all the requests for the stored contents. It is clear that a lower bound on the transmission rate in the new system is also a lower bound on the transmission rate of the original system.\\
In the new system, we can store at most $mk$ files. Therefore, all requests for the $n-mk$ files that are not stored have to be served by the central server. Therefore,
\begin{align*}
\mathbb{E}[R^*_{z_2}]&\geq \int_{n-mk+1}^{n} \Bigg(1-\Big(1-\frac{p_1}{i^\beta}\Big)^{r}\Bigg)di 
=\Omega\Big(n^{(2-\beta)}\Big).
\end{align*}
\textbf{Case 2:} $k = c$, $a=m^{\gamma}$: We use Proposition \ref{prop:zipf12} to prove this result. It can be shown that the optimal solution to $O^*$ has the following structure: $\exists$ $i_{\min} \geq 1$ and $i_{\max} \leq n$, such that, $x_i = 1$ if $i_{\text{min}} < {{i}} < i_{\text{max}}$, $x_{i_{\text{min}}} = f_1$ where $0\leq f_1 \leq 1$, $x_{i_{\text{max}}} = f_2$ where $0\leq f_2 \leq 1$, and $x_i = 0$ otherwise.   
%
Let $\widetilde i=\Big\lceil\big(\frac{rp_1}{2a}\big)^{\frac{1}{\beta}}\Big\rceil$. {By the definition of the fractional Knapsack problem}, 
\begin{align}
& f_1\frac{rp_1}{ai_{\min}^{\beta}}+\sum_{i=i_{\min}+1}^{\widetilde i-1}\Big\lfloor\frac{rp_i}{a}\Big\rfloor+\sum_{i=\widetilde i}^{i_{\max}-1}1+f_2 = mk, \nonumber 
\end{align}
\begin{align}
\therefore  i_{max}\leq mk+3\widetilde{i}&-f_1\frac{rp_1}{ai_{\min}^{\beta}} + \frac{rp_1}{a(\beta-1)}\times \nonumber \\
&\bigg[-(i_{\min}+1)^{(-\beta+1)}+{(\widetilde i-1)}^{(-\beta+1)}\bigg].\nonumber  
\end{align}

Let $i_{\min}=m^{\alpha}$. Recall that the fractional Knapsack solution has $i_{\min}\leq \widetilde{i}$. Hence, $\alpha \leq \frac{1-\gamma}{\beta}$. If $\alpha < \frac{1-\gamma}{\beta}$, $i_{\max}=n(1-\oo(1))$ and,  $\frac{i_{\max}}{n} =1-c_1m^{-\alpha \beta+\alpha-\gamma}(1-\oo(1))$ for some $c_1>0$. {Let $\EE[R_1^*]$ denote the expected number of contents requested at least once that are more popular than Content $i_{\min}$. } By Lemma \ref{lemma:Zipf_d} Part (a), $\EE[R_1^*]=m^{\alpha}$.
Let $\EE[R_2^*]$ denote the expected number of contents requested at least once that are less popular than Content $i_{\max}$.
\begin{align}
\EE[R_{2}^*]&\geq \int_{i_{max}+1}^{n}\frac{{r}p_1}{i^\beta}di 
=\Omega\Big(m^{-\beta +2-\alpha \beta+\alpha-\gamma}\Big). \nonumber 
\end{align}
\begin{align}
\therefore \EE[R_{z_2}^*] &\geq \EE[R_{1}^*]+\EE[R_{2}^*] 
\geq \Omega\Big(n^{\frac{2-\beta-\gamma}{\beta}}\Big).\nonumber
\end{align}
If $\alpha=\frac{1-\gamma}{\beta}$, and $i=\oo(m^{\frac{1-\gamma}{\beta}})$ then Content $i$ is not cached. From Lemma \ref{lemma:Zipf_d} Part (a), all these files are requested at least once. Hence, $\forall$ $\epsilon>0$,
$\mathbb{E}[R^*_{z_2}] \geq \Omega\Big(m^{\frac{1-\gamma}{\beta}-\epsilon}\Big),$
i.e.,
$
\mathbb{E}[R^*_{z_2}]  \geq \Omega\Big(n^{\frac{2-\beta-\gamma}{\beta}}\Big). \nonumber
$\\
\textbf{Case 3:} $k-c>\Theta(1)$ --The bound of 0 follows trivially.
 \end{proof}
\subsection{Proof of Theorem \ref{thm:zipf12u}}
We use the following lemmas in the proof of Theorem \ref{thm:zipf12u}. {These lemmas tell that if a file is stored in KS + MLP policy, then all its  requests are served by the caches with high probability.}

\begin{lemma}
	\label{lemma:serve_all_requests}
	Let $\mathcal{R} = \{i: x_i = 1\}$, where $x_i$ is the solution of the fraction Knapsack problem solved in Knapsack Storage: Part 1. Let $E_2$ be the event that the Match Least Popular policy matches all requests for all contents in $R$ to caches. Then, 
	$\PP(E_2) = 1-\OO(m e^{-3\log m}).$
\end{lemma}
\begin{proof}
	{
	Since the Match Least Popular policy matches requests to caches starting from the least popular contents, we first focus on requests for contents less popular than Content $n_2$. Since content popularity follows the Zipf distribution with Zipf parameter $\beta>1$, for $i > n_2$,
	$p_i	<p_{n_2} = \dfrac{p_1 }{m^{(1+\delta)}}.$
	 Since each cache stores at most $a{k}$ contents, the cumulative popularity of all contents less popular than Content $n_2$ stored on a cache is $< a{k}p_{n_2}$. Let $\widetilde{X}$ denote the number of requests for a cache for the contents with index greater than $n_2$. Then, $\EE[\widetilde{X}]=\OO\left(\frac{ak}{m^{\delta}}\right),$ and
	$\PP(\widetilde{X}\geq a)\leq \OO\left(\Big(\frac{e}{m^{\delta}}\Big)^{a}\right).$
	Since, each content is stored on $\big\lceil\frac{4}{a\delta}\big\rceil$ caches, the probability of a content with index greater than $n_2$ being unmatched is $\OO\left(\Big(\frac{e}{m^{\delta}}\Big)^{a\lceil\frac{4}{a\delta}\rceil}\right)$.
	By the union bound, the probability that at least one request for Content $i \in R$ such that $i >n_2$ is not matched by the Match Least Popular policy is $\leq  \OO\left(m\Big(\frac{e}{m^{\delta}}\Big)^{a\lceil\frac{4}{a\delta}\rceil}\right)=\OO\big(m{e^{-3\ln m}}\big)$.

	Next, we focus on contents ranked between $2$ and $n_2$. Note that, if the Knapsack Storage policy decides to store Content $i$, it stores it on $aw_i$ caches. 
	\begin{align} \label{eq:memoryn2}
	\sum_{i=2}^{n_2} x_ia w_i &\leq  \sum_{i=2}^{n_1} \bigg\lceil \bigg(1 + \dfrac{p_1}{2}\bigg) mp_i \bigg\rceil + \sum_{i=n_1+1}^{n_2} \lceil 4 p_1 (\log n)^2 \rceil \nonumber \\
	&\leq \big(1-\frac{p_1}{2}\big)m.
	\end{align}
	
	Therefore, if contents are stored according to Knapsack Storage: Part 2, each cache stores at most {one part of contents with index $i$ such that $2 \leq i \leq n_2$}. {We first focus on the contents ranked between $n_1$ and $n_2$. Let $D_{1,i}$ be the set of caches storing {parts of} Content $i$ for $n_1 \leq i \leq n_2$. Each content part is stored on $\big\lceil\frac{4 p_1 (\log n)^2}{a}\big\rceil$ caches. Let $E_{3,i}$ be the total number of requests from contents whose index is $>n_2$ and which are stored on Caches belonging to $D_{1,i}$. Hence,
		$\EE[E_{3,i}]\leq {r}\bigg\lceil\frac{4 p_1 (\log n)^2}{a}\bigg\rceil a{k}p_{n_2}.$ Therefore,
		$
		\PP(E_{3,i}\geq {2 p_1 (\log n)^2}) \leq \OO\bigg(\Big(\frac{1}{n^{\delta}}\Big)^{(\log n)^2}\bigg).
		$
		Hence, from Lemma \ref{lemma:Zipf_d}, the probability that Content $i$ for $n_1<i \leq n_2$, $i \in R$ is not served $\leq \OO\big({n^{-\delta{(\log n)^2}}}\big).$ By the union bound, the probability that contents belong to $R$ and ranked between $n_1$ and $n_2$ unmatched  to copies of the caches is $\leq \OO\big({n^{-\delta{(\log n)^2} +1}}\big)=\OO\big({e^{-3\ln m}}\big).$
		
		We next focus on the contents ranked between $2$ and $n_1$. Let $D_{2,i}$ be the set of caches storing {parts of} Content $i$ for $2 \leq i \leq n_1$. Each content part is stored on $\big\lceil\frac{(1+\frac{p_1}{2})mp_i}{a}\big\rceil$. Let $E_{4,i}$ be the total number of requests from contents whose index is $>n_2$ and which are stored in caches belong to $D_{2,i}$. Hence,
		$\EE[E_{4,i}]\leq {r}\bigg\lceil\frac{(1+\frac{p_1}{2})mp_i}{a}\bigg\rceil a{k}p_{n_2}.$ Therefore,
		$
		\PP(E_{4,i}\geq \frac{p_1}{4}mp_i)  \leq \OO\Big(\big(\frac{1}{n^{\delta}}\big)^{mp_i}\Big) \leq \OO\Big(\big(\frac{1}{n^{\delta}}\big)^{(\log n)^2}\Big).
		$
		Hence, from Lemma \ref{lemma:Zipf_d}, the probability that Content $i$  for $2\leq i \leq n_1$, $i \in R$ is not served $\leq \OO\big({n^{-\delta{(\log n)^2}}}\big).$ By the union bound, the probability that contents belong to $R$ and ranked between $2$ and $n_1$ unmatched is $\leq \OO\big({n^{-\delta{(\log n)^2} +1}}\big)=\OO\big({e^{-3\ln m}}\big).$
	}
	
}
	Finally, We focus on the requests for Content 1. Recall that if the Knapsack Storage policy decides to cache Content 1, it is stored on  $\frac{m}{a}$ caches. Since the total number of requests in a batch is ${r}$, even if all requests  for contents ranked lower than 1 are matched to caches, the remaining caches can be used to serve all the requests for Content 1.
\end{proof}

The next lemma evaluates the performance of the Knapsack Store + Match Least Popular (KS+MLP) policy for the case where content popularity follows the Zipf distribution. 

\begin{lemma}
	\label{lemma:performance:our_policy}
	Consider a distributed cache consisting of a central server and $m$ caches that offers a catalog of $n$ contents. Let a batch of ${r}$ requests arrive in each time-slot and $R_{\text{KS+MLP}}$ be the transmission rate for the KS+MLP policy when content popularity follows the Zipf distribution with Zipf parameter $\beta>1$. 
	Then, we have that, for $n$ large enough,	
	$
		\EE[R_{\text{KS+MLP}}] \leq \sum_{i \notin R} 1-\bigg(1-\frac{p_1}{i^{\beta}}\bigg)^{{r}} + \OO(n^2 e^{-3\log n}),
$
	where $p_1 = \big(\sum_{i=1}^n i^{-\beta}\big)^{-1}$, $\mathcal{R} = \{i: x_i = 1\}$, such that $x_i$ is the solution of the fraction Knapsack problem solved in Knapsack Storage: Part 1.
\end{lemma}
\begin{proof}
	From Lemma \ref{lemma:serve_all_requests}, we know that, for $n$ large enough, with probability $\geq 1-\OO\big(me^{-3\log m}\big)$, all requests for the contents cached by the KS+MLP policy are matched to caches. Let $\tilde{n}$ be the number of contents not in $\mathcal{R}$ (i.e., not cached by the KS+MLP policy) that are requested at least once in a given time-slot. Therefore, 
	$\EE[\tilde{n}] = \sum_{i \notin \mathcal{R}} 1-(1-p_1)^{{r}}, \text{ and, }$
	$\text{ }\hspace{0.5in} \EE[R_{\text{KS+MLP}}] \leq \EE[\tilde{n}] P(E_2) + m (1-P(E_2))$\\ $\text{ }\hspace{1.16in}\leq\EE[\tilde{n}] + \OO(m^2 e^{-3\log m}). $
\end{proof}

\begin{proof} (Proof of Theorem \ref{thm:zipf12u}) Let $R_{z_2}$ denote the number of files that are not stored by the KS+MLP policy and are requested at least once.\\  
	\textbf{Case 1:} $c-k>\Theta(1)$: From Lemma \ref{lemma:serve_all_requests}, if we store $w_i$ times Content $i$ according to the Knapsack Storage Policy: Part 2, all the requests for it are served with high probability. Let $R$ be the transmission rate of the policy which stores from File 2 onwards, each file $w_i$ times according to Knapsack Storage Policy: Part 2  until memory is full. From equation (\ref{eq:memoryn2}), we  store more than $\frac{p_1'}{2}m$ files $w_i$ times. From the definition of fractional Knapsack problem, $\EE[R_{z_2}] \leq \EE[R]$.\\
	$\text{ }\hspace{0.25in}\mathbb{E}[R]\leq 1+ \int_{\big(c-\frac{p_1'}{2}\big)m}^{cm}1-(1-p_i)^{{r}} di\approx \OO\Big(n^{2-\beta}\Big).$
	
	{	
	\textbf{Case 2:} $c=k$, $a=m^\gamma$: Let the Knapsack solution be store files from $i_{min}+1$ to $i_{max}$. 
	
	\begin{align}
\therefore \int_{i_{\min}}^{\frac{m^{\frac{1}{\beta}}}{(\log m)^\frac{2}{\beta}}}\bigg( \Big(1 + \frac{p_1}{2}\Big) \frac{{r}p_i}{a} \bigg) di+\bigg(  \frac{4 p_1(\log m)^2}{a} \bigg) n^{\frac{1+\delta}{\beta}}\nonumber\\
+\Big\lceil\frac{4}{a\delta}\Big\rceil i_{\max} \geq mk.\nonumber
	\end{align}
	
	\begin{align}
	\implies \Big\lceil\frac{4}{a\delta}\Big\rceil i_{\max}\geq m&k- \frac{4p_1\log m}{a}n^{\frac{1+\delta}{\beta}}+\Big(1+\frac{p_1}{2}\Big)\times\nonumber\\
&	\frac{mp_1}{a(\beta-1)}\Bigg[\Bigg(\frac{m^{\frac{1}{\beta}}}{(\log m)^\frac{2}{\beta}}\Bigg)^{1-\beta}-i_{\min}^{(1-\beta)}\Bigg]. \nonumber
	\end{align}
	Let $i_{\min}=m^{\alpha}$ for some $\alpha<\frac{1-\gamma}{\beta}$, and substitute it in the above equation, we get
	$i_{\max}=c_1n(1-\oo(1))$ and $\frac{i_{\max}}{n}=1-c_2\frac{m^{-\alpha(\beta-1)}}{a}(1-\oo(1))$ for some $c_1>0,$ $c_2>0$. From Lemma \ref{lemma:performance:our_policy},
		\begin{align*}
		\mathbb{E}[R_{z_2}]  &\leq i_{\min}+\int_{i_{\max}}^{n} 1-(1-p_i)^{{r}} di +\oo(1)\nonumber\\
		&= m^{\alpha}+\OO\Big(m^{2-\beta-\alpha(\beta-1)-\gamma}\Big).
	\end{align*} 
By taking  $\alpha=\frac{2-\beta-\gamma}{\beta}$, we will get
$
	\mathbb{E}[R_{z_2}]=\OO\Big(n^\frac{2-\beta-\gamma}{\beta}\Big).
$
	
{\textbf{Case 3:} {$k-c>\Theta(1)$}: Let $mk=n+lm$.
	If we store Contents $t$ to $n$, the total memory required is less than 
	\begin{align}\label{eqn:finalmemory}
	&\int_{t}^{\frac{m^{\frac{1}{\beta}}}{(\log m)^\frac{2}{\beta}}}\bigg( \Big(1 + \frac{p_1}{2}\Big) \frac{{r}p_i}{a} \bigg) di+\bigg(  \frac{4 p_1(\log m)^2}{a} \bigg) n^{1-\epsilon}+\nonumber \\
	&\Big\lceil\frac{4}{a\delta}\Big\rceil n 
	\leq \Big(1+\frac{p_1}{2}\Big)\frac{mp_1}{a(\beta-1)}\Bigg[t^{(1-\beta)}-\Bigg(\frac{m^{\frac{1}{\beta}}}{(\log m)^\frac{2}{\beta}}\Bigg)^{1-\beta}\Bigg] \nonumber \\
	&+\bigg(  \frac{4 p_1(\log m)^2}{a} \bigg) n^{1-\epsilon}+\Big\lceil\frac{4}{a\delta}\Big\rceil n.\nonumber
	\end{align}
	For $l\geq \big\lceil\frac{4}{a\delta}\big\rceil$, $\exists t$, such that the total memory required  is less than $mk$. Therefore,
	$\EE[R_{z_2}]\leq t=\Theta(1).$}	
}	
\end{proof}

		\bibliographystyle{IEEEtran}
	\bibliography{myref2}

\begin{thebibliography}{10}
\providecommand{\url}[1]{#1}
\csname url@samestyle\endcsname
\providecommand{\newblock}{\relax}
\providecommand{\bibinfo}[2]{#2}
\providecommand{\BIBentrySTDinterwordspacing}{\spaceskip=0pt\relax}
\providecommand{\BIBentryALTinterwordstretchfactor}{4}
\providecommand{\BIBentryALTinterwordspacing}{\spaceskip=\fontdimen2\font plus
\BIBentryALTinterwordstretchfactor\fontdimen3\font minus
  \fontdimen4\font\relax}
\providecommand{\BIBforeignlanguage}[2]{{%
\expandafter\ifx\csname l@#1\endcsname\relax
\typeout{** WARNING: IEEEtran.bst: No hyphenation pattern has been}%
\typeout{** loaded for the language `#1'. Using the pattern for}%
\typeout{** the default language instead.}%
\else
\language=\csname l@#1\endcsname
\fi
#2}}
\providecommand{\BIBdecl}{\relax}
\BIBdecl

\bibitem{moharir2017content}
S.~Moharir and N.~Karamchandani, ``Content replication in large distributed
  caches,'' in \emph{IEEE 9th International Conference on Communication Systems
  and Networks (COMSNETS)}, 2017, pp. 128--135.

\bibitem{reddy2017resource}
K.~S. Reddy, S.~Moharir, and N.~Karamchandani, ``Resource pooling in
  large-scale content delivery systems,'' in \emph{IEEE Twenty-third National
  Conference on Communications (NCC)}, 2017, pp. 1--6.

\bibitem{Youtube}
YouTube: {\tt http://www.youtube.com}.

\bibitem{Cisco1}
Cisco Whitepaper: {\tt https://www.cisco.com/c/en/us/soluti
  ons/collateral/service-provider/visual-networking
  -index-vni/white-paper-c11-741490.html}.

\bibitem{borst2010distributed}
S.~Borst, V.~Gupt, and A.~Walid, ``Distributed caching algorithms for content
  distribution networks,'' in \emph{IEEE Conference on Computer Communications
  (INFOCOM)}, 2010, pp. 1--9.

\bibitem{SGSS14}
S.~Moharir, J.~Ghaderi, S.~Sanghavi, and S.~Shakkottai, ``Serving content with
  unknown demand: the high-dimensional regime,'' in \emph{ACM SIGMETRICS},
  2014.

\bibitem{LLM12}
M.~Leconte, M.~Lelarge, and L.~Massouli{\'e}, ``Bipartite graph structures for
  efficient balancing of heterogeneous loads,'' in \emph{ACM SIGMETRICS}, 2012,
  pp. 41--52.

\bibitem{LLM13}
------, ``Designing adaptive replication schemes in distributed content
  delivery networks,'' in \emph{Teletraffic Congress (ITC 27)}, 2015, pp.
  28--36.

\bibitem{Whitt07}
R.~B. Wallace and W.~Whitt, ``A staffing algorithm for call centers with
  skill-based routing,'' \emph{Manufacturing and Service Operations
  Management}, vol.~7, pp. 276--294, 2007.

\bibitem{XT13}
J.~Tsitsiklis and K.~Xu, ``Queueing system topologies with limited
  flexibility,'' in \emph{ACM SIGMETRICS}, 2013.

\bibitem{maddah2014fundamental}
M.~Maddah-Ali and U.~Niesen, ``Fundamental limits of caching,'' \emph{IEEE
  Transactions on Information Theory}, vol.~60, no.~5, pp. 2856--2867, 2014.

\bibitem{maddah2013decentralized}
------, ``Decentralized coded caching attains order-optimal memory-rate
  tradeoff,'' in \emph{Annual Allerton Conference on Communication, Control,
  and Computing (Allerton)}, 2013, pp. 421--427.

\bibitem{pedarsani2014online}
R.~Pedarsani, M.~Maddah-Ali, and U.~Niesen, ``Online coded caching,'' in
  \emph{IEEE International Conference on Communications (ICC)}, 2014, pp.
  1878--1883.

\bibitem{niesen2014coded}
U.~Niesen and M.~Maddah-Ali, ``Coded caching with nonuniform demands,'' in
  \emph{IEEE Conference on Computer Communications Workshops (INFOCOM WKSHPS)},
  2014, pp. 221--226.

\bibitem{hachem2014multi}
J.~Hachem, N.~Karamchandani, and S.~Diggavi, ``Multi-level coded caching,'' in
  \emph{IEEE International Symposium on Information Theory (ISIT)}, 2014.

\bibitem{zhang2015coded}
J.~Zhang, X.~Lin, and X.~Wang, ``Coded caching under arbitrary popularity
  distributions,'' in \emph{IEEE Information Theory and Applications (ITA)
  Workshop}, 2015, pp. 98--107.

\bibitem{shanmugam2013femtocaching}
K.~Shanmugam, N.~Golrezaei, A.~Dimakis, A.~Molisch, and G.~Caire,
  ``Femtocaching: Wireless content delivery through distributed caching
  helpers,'' \emph{IEEE Transactions on Information Theory}, vol.~59, no.~12,
  pp. 8402--8413, 2013.

\bibitem{wessels2001}
D.~Wessels, \emph{Web Caching}, N.~Torkington, Ed.\hskip 1em plus 0.5em minus
  0.4em\relax O'Reilly, 2001.

\bibitem{tan2013optimal}
B.~Tan and L.~Massouli{\'e}, ``Optimal content placement for peer-to-peer
  video-on-demand systems,'' \emph{IEEE/ACM Transactions on Networking (TON)},
  vol.~21, no.~2, pp. 566--579, 2013.

\bibitem{paschos2018role}
G.~S. Paschos, G.~Iosifidis, M.~Tao, D.~Towsley, and G.~Caire, ``The role of
  caching in future communication systems and networks,'' \emph{IEEE Journal on
  Selected Areas in Communications}, vol.~36, no.~6, pp. 1111--1125, 2018.

\bibitem{krishnan2013video}
S.~S. Krishnan and R.~K. Sitaraman, ``Video stream quality impacts viewer
  behavior: inferring causality using quasi-experimental designs,''
  \emph{IEEE/ACM Transactions on Networking}, vol.~21, no.~6, pp. 2001--2014,
  2013.

\bibitem{shah2014performance}
V.~Shah and G.~de~Veciana, ``Performance evaluation and asymptotics for content
  delivery networks,'' in \emph{IEEE Conference on Computer Communications
  (INFOCOM)}, 2014, pp. 2607--2615.

\bibitem{shah2015high}
------, ``High-performance centralized content delivery infrastructure: models
  and asymptotics,'' \emph{IEEE/ACM Transactions on Networking}, vol.~23,
  no.~5, pp. 1674--1687, 2015.

\bibitem{reddy2018effects}
K.~S. Reddy, S.~Moharir, and N.~Karamchandani, ``Effects of storage
  heterogeneity in distributed cache systems,'' in \emph{IEEE 16th
  International Symposium on Modeling and Optimization in Mobile, Ad Hoc, and
  Wireless Networks (WiOpt)}, 2018, pp. 1--8.

\bibitem{liu2013measurement}
Y.~Liu, F.~Li, L.~Guo, B.~Shen, S.~Chen, and Y.~Lan, ``Measurement and analysis
  of an internet streaming service to mobile devices,'' \emph{IEEE Transactions
  on Parallel and Distributed Systems}, vol.~24, no.~11, pp. 2240--2250, 2013.

\bibitem{BC99}
L.~Breslau, P.~Cao, L.~Fan, G.~Phillips, and S.~Shenker, ``Web caching and
  {Zipf}-like distributions: Evidence and implications,'' in \emph{IEEE
  Conference on Computer Communications (INFOCOM)}, 1999, pp. 126--134.

\bibitem{YZ06}
H.~Yu, D.~Zheng, B.~Zhao, and W.~Zheng., ``Understanding user behavior in large
  scale video-on-demand systems,'' in \emph{EuroSys}, 2006.

\bibitem{fricker2012impact}
C.~Fricker, P.~Robert, J.~Roberts, and N.~Sbihi, ``Impact of traffic mix on
  caching performance in a content-centric network,'' in \emph{IEEE Conference
  on Computer Communications Workshops (INFOCOM WKSHPS)}, 2012, pp. 310--315.

\bibitem{goodrich2006algorithm}
M.~T. Goodrich and R.~Tamassia, \emph{Algorithm design: foundation, analysis
  and internet examples}.\hskip 1em plus 0.5em minus 0.4em\relax John Wiley \&
  Sons, 2006.

\end{thebibliography}
	
	


\end{document}